\documentclass[acmsmall,screen,nonacm,dvipsnames,svgnames]{acmart}

\let\normalunderbrace=\underbrace
\usepackage{mathtools} 
\let\underbrace=\normalunderbrace
\usepackage{hyperref}
\usepackage{enumitem}
\usepackage{stmaryrd}  
\usepackage{apxproof}  
\usepackage{graphicx}  
\usepackage{tikz-cd}
\usepackage{circuitikz}

\usepackage{tikz}
\usetikzlibrary{positioning}

\providecommand\strong[1]{{\bfseries#1}}

\usepackage[nameinlink,capitalize]{cleveref}
\newcommand*{\fullref}[1]{\cref{#1} (\nameref*{#1})}

\newcommand*{\propref}[1]{\hyperref[#1]{Proposition \ref*{#1}}}
\newcommand*{\Propref}[1]{\hyperref[#1]{Proposition \ref*{#1}}}
\newcommand*{\lemref}[1]{\hyperref[#1]{Lemma \ref*{#1}}}
\newcommand*{\Lemref}[1]{\hyperref[#1]{Lemma \ref*{#1}}}

\AtEndPreamble{
	\newtheoremrep{theoremapx}[theorem]{Theorem}
	\newtheoremrep{propositionapx}[theorem]{Proposition}
	\newtheoremrep{lemmaapx}[theorem]{Lemma}
	\theoremstyle{acmdefinition}

	\newtheorem*{remark}{Remark}
	
}
\crefformat{procedure}{procedure}
\Crefformat{procedure}{Procedure}
\crefformat{proof}{proof}
\Crefformat{proof}{Proof}

\showboxdepth=\maxdimen
\showboxbreadth=\maxdimen

\DeclareMathOperator{\State}{State}
\DeclareMathOperator{\out}{out}
\DeclareMathOperator{\st}{state}
\DeclareMathOperator{\List}{List}

\DeclareMathOperator{\run}{run}
\DeclareMathOperator{\pure}{pure}
\DeclareMathOperator{\eval}{eval}

\DeclareMathOperator{\loopfn}{loop}

\DeclareMathOperator{\idfn}{id}
\DeclareMathOperator{\fmap}{fmap}

\DeclareMathOperator{\prefixSum}{PS}
\DeclareMathOperator{\sender}{sender}
\DeclareMathOperator{\network}{net}
\DeclareMathOperator{\receiver}{receiver}

\DeclareMathOperator*{\odots}{\odot}
\DeclarePairedDelimiter{\denote}{\llbracket}{\rrbracket}

\let\oldState\State
\renewcommand{\State}[2]{\oldState[#1, #2]}

\newcommand{\Set}[1]{\mathcal{P}(#1)}
\newcommand{\Img}[1]{\mathrm{Im}(#1)}

\newcommand{\defeq}{\triangleq}
\newcommand{\pipe}{\mathbin{;}}
\newcommand{\compose}{\circ}
\newcommand{\dprod}{\times}
\newcommand{\tensor}{\otimes}
\newcommand{\kleisli}{\mathrel{>\hspace{-0.23em}=\hspace{-0.35em}>\!}}
\newcommand{\doubleplus}{\mathbin{+\hspace{-0.4em}+}}

\newcommand{\processor}[1]{\sigma_{#1}}
\newcommand{\processortwo}[1]{\tau_{#1}}
\newcommand{\processorthree}[1]{\upsilon_{#1}}

\newcommand{\psplit}{\mathrm{split}}
\newcommand{\pmerge}{\mathrm{merge}}
\newcommand{\pairs}{\mathrm{pairs}}
\newcommand{\filter}{\mathrm{filter}}

\begin{document}

\title{Stream programs are monoid homomorphisms with state} 

\author{Tyler Hou}
\email{tylerhou@berkeley.edu}
\orcid{0009-0002-2234-2549}
\affiliation{%
  \institution{UC Berkeley}
  \city{Berkeley}
  \state{CA}
  \country{USA}
}

\author{Michael Arntzenius}
\email{daekharel@gmail.com}
\orcid{0009-0002-0417-5636}
\affiliation{%
  \institution{UC Berkeley}
  \city{Berkeley}
  \state{CA}
  \country{USA}
}

\author{Max Willsey}
\email{mwillsey@berkeley.edu}
\orcid{0000-0001-8066-4218}
\affiliation{%
  \institution{UC Berkeley}
  \city{Berkeley}
  \state{CA}
  \country{USA}
}

\begin{abstract}
  \noindent
  We define a broad class of deterministic stream functions and show they can be
  implemented as homomorphisms into a ``state'' monoid. The homomorphism laws are
  simpler than the conditions of previous semantic frameworks for stream program
  optimization, yet retain support for rich equational reasoning over expressive
  dataflow programs, including sequential composition, parallel composition, and
  feedback. We demonstrate this using examples of partitioned database joins,
  stratified negation, and a simplified model of TCP.
\end{abstract}


\maketitle

\section{Introduction}

Stream programs receive their inputs and produce their outputs incrementally
over time. This is important for programs which must react to external change,
as well as in parallel and distributed programming where it is important to do
work as soon as possible rather than delaying until all input is available.
However, reasoning about and optimizing stream programs is an active area of
research.
It is common to model data streams as
monoids~\citep{mamouras,DBLP:journals/pacmpl/LaddadCHM25} (or richer algebraic
structures, e.g.\ groups~\citep{DBLP:journals/pvldb/BudiuCMRT23}), and stream
programs as operators that transform these monoids. It is well-known that
stateless stream programs can be modeled as monoid homomorphisms; this is the
essential insight of MapReduce~\citep{DBLP:conf/osdi/DeanG04}. However, many
stream programs (e.g., joins, aggregates) are stateful, and thus cannot be
modeled as monoid homomorphisms in a straightforward way.

In this paper, we propose a formalism that reasons about both stateless and
stateful stream programs using monoid homomorphisms. We show that stateful
stream programs can be viewed as a homomorphism into a ``State'' monoid. The
State monoid contains elements which are functions of type \( S \to S \times N
\), where \( S \) is a state type and \( N \) is the output monoid. We show
monoid homomorphisms \emph{into} this State monoid represent possibly-stateful
stream programs.

This key insight allows the decomposition of stateful stream programs into
their pure and stateful parts. We prove several theorems that allow one to
compose, decompose, parallelize, and rearrange parts of stream programs
(\cref{sec:semantics}). We demonstrate the utility of our formalism with
examples of
partitioning database joins, stratified negation, and a simplified model of
TCP. We also show a connection to compilers techniques such as
defunctionalization and partial evaluation, and show how to use this connection
to optimize stream programs (\cref{sec:partial-eval}).

\section{Motivating example: optimizing joins}
\label{sec:example}

A key problem in distributed database systems is optimizing joins. Let \( E \)
be a relation with arity two representing directed edges in a graph. Suppose we
want to compute the paths of length two: that is given by the join \( Q(x, y,
z) = E(x, y) \wedge E(y, z) \) (the repeated variable \( y \) denotes that the
first edge's target must be the same as the second edge's source).

\subsection{From pairs-filter to join}

The naïve way to compute such a join is to enumerate \( E \tensor E
\)\footnote{ For now, one can think of \( E \tensor E \) as the Cartesian
  product of \( E \) with itself. In \cref{sec:parallel}, we define the tensor
  product \( \tensor \). } and filter out any pairs whose \( y \) entries do not
match. Graphically, that would be the stream program%

\begin{figure}[!ht]
  \centering
  \begin{tikzpicture}
    \node[draw] (pairs) {pairs};
    \node[above left=0cm and 1cm of pairs] (A) {\( E \)};
    \node[below left=0cm and 1cm of pairs] (B) {\( E \)};
    \node[draw, right=of pairs] (filter) {pure filter};
    \node[right=of filter] (out) {};

    \draw (A) edge[->] (pairs);
    \draw (B) edge[->] (pairs);
    \draw (pairs) edge[->, "\( E \tensor E \)"] (filter);
    \draw (filter) edge[->, "\( E \tensor E \)"] (out);
  \end{tikzpicture}
\end{figure}

\noindent
where \( \filter \) is ``pure'' because it operates \textit{statelessly} over
its input.

The above program is not very efficient. The first operator must output
\textit{all pairs} of the join, only for the second operator to filter many of
them out. As humans, it clear to us that this composition is equivalent to a
single operator which enumerates pairs and only emits ones which pass the
filter. How might we enable a compiler to discover this optimization
automatically?

As we previously observed, \( \pure \filter \) is a stateless operator; i.e.,
it only needs to examine a single pair at a time to determine whether it should
be discarded or not. If \( \pairs \) was also a stateless operator, then we
could fuse the two operators together. The fused operator would examine one
input element at a time and only emit pairs that pass the filter.

Unfortunately, \( \pairs \) is not stateless---it needs to remember all
elements of both input streams to enumerate all pairs. However, we show in
\cref{cor:decompose} that we can decompose it into a stateless operator (\(
\pure f_{\pairs} \)) followed by an evaluation (\( \eval \pairs \)). The
outputs of \( \pure f_{\pairs} \) are ``hypotheticals'' over state. Given an
input \( (\delta E_1, \delta E_2) \), the processor \( \pure f_{\pairs} \)
emits the \textit{function}
\[
  (E_1, E_2) \mapsto ((E_1 + \delta E_1, E_2 + \delta E_2), \hspace{0.3em} E_1
  \tensor \delta E_1 + \delta E_1 \tensor E_2 + \delta E_1 \tensor \delta E_2)
\]
We interpret this function as: ``If \( (E_1, E_2) \) is \( \pairs \)'s current
state, update the state to \( (E_1 + \delta E_1, E_2 + \delta E_2) \) and emit
output \( E_1 \tensor \delta E_1 + \delta E_1 \tensor E_2 + \delta E_1 \tensor
\delta E_2 \).''

We show that these functions form a monoidal structure we call \( \State{S}{E
  \tensor E} \) (\cref{def:state}), and indeed \( f_{\pairs} : E \dprod E \to
\State{S}{E \tensor E} \) is a homomorphism! Then we obtain the program:

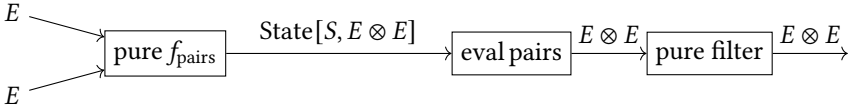
\begin{figure}[!ht]
  \centering
  \begin{tikzpicture}
    \node[draw] (pairs) {\( \pure f_\pairs \)};
    \node[above left=0cm and 1cm of pairs] (A) {\( E \)};
    \node[below left=0cm and 1cm of pairs] (B) {\( E \)};
    \node[draw, right=3cm of pairs] (eval) {\( \eval \pairs \)};
    \node[draw, right=of eval] (filter) {pure filter};
    \node[right=of filter] (out) {};

    \draw (A) edge[->] (pairs);
    \draw (B) edge[->] (pairs);
    \draw (pairs) edge[->, "\( \State{S}{E \tensor E} \)"] (eval);
    \draw (eval) edge[->, "\( E \tensor E \)"] (filter);
    \draw (filter) edge[->, "\( E \tensor E \)"] (out);
  \end{tikzpicture}
\end{figure}

\noindent
Because \( \pure \filter \) is a stateless operator and \( \State{S}{E \tensor
  E} \) is a functor in its second argument, we can \( \mathrm{fmap}\, \filter \)
into \( \State{S}{E \tensor E} \) to produce a function \( \filter_* :
\State{S}{E \tensor E} \to \State{S}{E \tensor E} \) which filters the output
``inside'' each function that \( \pure f_{\pairs} \) produces. Then \( \pure
\filter_* \) is also a stateless operator, and we can \textit{exchange} it
before \( \eval \pairs \) (\cref{cor:exchangeeval}), producing the program

\begin{figure}[!ht]
  \centering
  \begin{tikzpicture}
    \node[draw] (pairs) {\( \pure f_\pairs \)};
    \node[above left=0cm and 1cm of pairs] (A) {\( E \)};
    \node[below left=0cm and 1cm of pairs] (B) {\( E \)};
    \node[draw, right=2.5cm of pairs] (eval) {\( \pure \filter_* \)};
    \node[draw, right=2.5cm of eval] (filter) {\( \eval_* \pairs \)};
    \node[right=of filter] (out) {};

    \draw (A) edge[->] (pairs);
    \draw (B) edge[->] (pairs);
    \draw (pairs) edge[->, "\( \State{S}{E \tensor E} \)"] (eval);
    \draw (eval) edge[->, "\( \State{S}{E \tensor E} \)"] (filter);
    \draw (filter) edge[->, "\( E \tensor E \)"] (out);
  \end{tikzpicture}
\end{figure}

\noindent
Because \( \pure f_{\pairs} \) and \( \pure \filter \) are now both stateless,
we can fuse them together to produce an operator (\( \pure\; (f_\pairs \pipe
\filter_*) \)). This operator outputs functions that, when evaluated, enumerate
all pairs and only emit the ones which pass the filter. Thus we have the
program

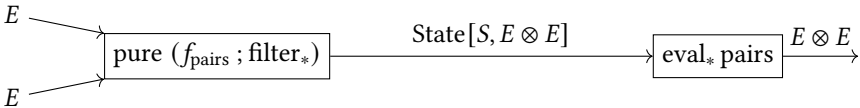
\begin{figure}[!ht]
  \centering
  \begin{tikzpicture}
    \node[draw] (pairs) {\( \pure\; (f_\pairs \pipe \filter_*) \)};
    \node[above left=0cm and 1cm of pairs] (A) {\( E \)};
    \node[below left=0cm and 1cm of pairs] (B) {\( E \)};
    \node[draw, right=0.7cm of eval] (filter) {\( \eval_* \pairs \)};
    \node[right=of filter] (out) {};

    \draw (A) edge[->] (pairs);
    \draw (B) edge[->] (pairs);
    \draw (pairs) edge[->, "\( \State{S}{E \tensor E} \)"] (filter);
    \draw (filter) edge[->, "\( E \tensor E \)"] (out);
  \end{tikzpicture}
\end{figure}

\noindent
Finally, we can fuse the composition \( \pure\; (f_\pairs \pipe \filter_*) \)
into a single operator (\propref{prop:fusion}), \( \mathrm{join} \), which
emits filtered pairs.

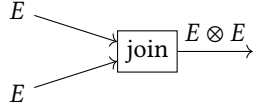
\begin{figure}[!ht]
  \centering
  \begin{tikzpicture}
    \node[draw] (join) {\( \mathrm{join} \)};
    \node[above left=0cm and 0cm of pairs] (A) {\( E \)};
    \node[below left=0cm and 0cm of pairs] (B) {\( E \)};
    \node[right=of join] (out) {};

    \draw (A) edge[->] (join);
    \draw (B) edge[->] (join);
    \draw (join) edge[->, "\( E \tensor E \)"] (out);
  \end{tikzpicture}
\end{figure}

\subsection{Join partitioning}

A common join optimization in distributed database implementations is to
\textit{partition} joins on the join key. We can partition \( E \) into four
partitions \
\begin{align*}
  E_{-,0} = \{ (s, t) \in E \mid t \equiv 0 \pmod 2 \}
  \hspace{2em} &
  E_{0,-} = \{ (s, t) \in E \mid s \equiv 0 \pmod 2 \}
  \\
  E_{-,1} = \{ (s, t) \in E \mid t \equiv 1 \pmod 2 \}
  \hspace{2em} &
  E_{1,-} = \{ (s, t) \in E \mid s \equiv 1 \pmod 2 \}
\end{align*}

\noindent
Then an equivalent query is \( Q(x, y, z) = (E_{-, 0}(x, y) \wedge E_{0, -}(y,
z)) \vee (E_{-, 1}(x, y) \wedge E_{1, -}(y, z)) \). Observe that we need not
compare tuples in \( E_{-,0} \wedge E_{1,-} \) and \( E_{-,0} \wedge E_{1,-}
\), as they must produce empty results since \( y \)'s parity does not match.

To apply this optimization, we start with the stream program above containing
just the \( \mathrm{join} \) operator.

\begin{figure}[!ht]
  \centering
  \begin{tikzpicture}
    \node[draw] (join) {\( \mathrm{join} \)};
    \node[above left=0cm and 0cm of pairs] (A) {\( E \)};
    \node[below left=0cm and 0cm of pairs] (B) {\( E \)};
    \node[right=of join] (out) {};

    \draw (A) edge[->] (join);
    \draw (B) edge[->] (join);
    \draw (join) edge[->, "\( E \tensor E \)"] (out);
  \end{tikzpicture}
\end{figure}

We introduce operators \( \pure \psplit \) and \( \pure \pmerge \) on the left.
The operator \( \pure \psplit \) splits the input tuples according to the
partitioning scheme above. We define these operators such that their
composition is the identity (\cref{def:splitter}), so this does not change the
semantics of the program.

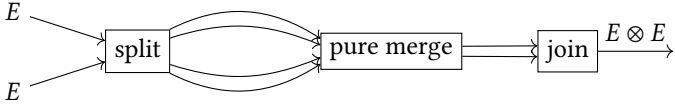
\begin{figure}[!ht]
  \begin{tikzpicture}
    \node[draw] (split) {split};
    \node[above left=0cm and 1cm of split] (A) {\( E \)};
    \node[below left=0cm and 1cm of split] (B) {\( E \)};
    \node[draw, right=2cm of split] (merge) {pure merge};
    \node[draw, right=of merge] (join) {join};
    \node[right=of join] (out) {};

    \draw (A) edge[->] (split);
    \draw (B) edge[->] (split);
    \draw (split) edge[->, bend left=30, "" above, near start] (merge.170);
    \draw (split) edge[->, bend left=20, "" above, near start] (merge.175);
    \draw (split) edge[->, bend right=20, "" below, near start] (merge.185);
    \draw (split) edge[->, bend right=30, "" below, near start] (merge.190);
    \draw (merge.4) edge[->, ""] (join.170);
    \draw (merge.-4) edge[->, ""] (join.190);
    \draw (join) edge[->, "\( E \tensor E \)"] (out);
  \end{tikzpicture}
\end{figure}

Because of the way we defined \( \pure \psplit \) to partition the input, we
can fully evaluate \( \mathrm{join} \) on both partitions in parallel
(\cref{def:independent}). In such cases, we can commute \( \mathrm{join} \)
left of \( \pure \pmerge \) (\lemref{lem:independentmerge}), producing the
final optimized join (\cref{thm:join-partition}):

\begin{figure}[!ht]
  \begin{tikzpicture}
    \node[draw] (split) {split};
    \node[above left=0cm and 1cm of split] (A) {\( E \)};
    \node[below left=0cm and 1cm of split] (B) {\( E \)};
    \node[draw, above right=0cm and 1cm of split] (join1) {join};
    \node[draw, below right=0cm and 1cm of split] (join2) {join};
    \node[draw, below right=0cm and 1cm of join1] (merge) {pure merge};
    \node[right=of merge] (out) {};

    \draw (A) edge[->] (split);
    \draw (B) edge[->] (split);
    \draw (split) edge[->, bend left=10, ""] (join1);
    \draw (split) edge[->, bend right=10] (join1);
    \draw (split) edge[->, bend left=10] (join2);
    \draw (split) edge[->, bend right=10, "" below, near start] (join2);
    \draw (join1) edge[->, "\( E \tensor E \)" above, near end] (merge);
    \draw (join2) edge[->, "\( E \tensor E \)" below, near end] (merge);
    \draw (merge) edge[->, "\( E \tensor E \)"] (out);
  \end{tikzpicture}
\end{figure}

\section{Semantics}
\label{sec:semantics}

This section contains the key technical contributions of this paper. We begin
with brief preliminaries on monoids and homomorphisms (\cref{sec:prelim}). Then
we introduce the notions of \textit{stream functions}
(\cref{sec:stream-functions}), and \textit{stream processors}
(\cref{sec:stream-processor}), which serve respectively as specifications and
implementations of stream programs. We demonstrate that stream processors can
be composed in sequence (\cref{sec:sequential}), in parallel
(\cref{sec:parallel}), and with feedback loops (\cref{sec:loops}). Along the
way, we illustrate the utility of our formalism with examples of streaming
derivatives/integrals (\cref{sec:inverse}), partitioning of database joins
(\cref{sec:join}), ticked streams for stratified computation
(\cref{sec:ticks}), and a simplified model of TCP (\cref{sec:tcp}).

\setcounter{subsection}{-1}
\subsection{Preliminaries}
\label{sec:prelim}

\newcommand{\mul}{\cdot}
\newcommand{\id}{\varepsilon}
\newcommand{\Int}{\mathbb{Z}}
\newcommand{\Nat}{\mathbb{N}}

\begin{definition}
  A monoid \( (M, {\mul}, \id) \) is a set \( M \) equipped with a binary
  operator \( {\mul}: M \times M \to M \) called \textit{product} and element \(
  \id \in M \) called \textit{identity} satisfying, for all \( a,b,c \in M \): \
  \begin{align*}
    a \mul \id = a \hspace{-0.85em} & \hspace{0.85em} = \id \mul a \tag{identity} \\
    (a \mul b) \mul c               & = a \mul (b \mul c) \tag{associativity}
  \end{align*}

  \noindent
  We often write \( a \mul b \) as simply \( ab \). For example, associativity
  may be notated \( (ab)c = a(bc) \).
\end{definition}

\begin{definition}
  Let \( M \) be a monoid. An element \( a \in M \) is \textit{invertible} if
  there exists an element \( b \) such that \( ab = ba = \id \). We call \( b \)
  the inverse of \( a \), and we use notation \( a^{-1} \) to denote \( a \)'s
  inverse.
\end{definition}

\begin{definition}
  A group \( G \) is a monoid where every element \( g \in G \) is invertible.
\end{definition}

\begin{definition}
  Let \( M \) be a monoid. \( M \) is \textit{left-cancellative} if whenever \(
  xa = xb \) we have \( a = b \).
\end{definition}

Groups are left-cancellative, as we can simply left-multiply both sides of the
equation by \( x^{-1} \).

\begin{definition}
  Let \( M \) be a monoid and \( N \subseteq M \). We say that \( M \) is
  \textit{generated by \( N \)} if every element \( m \in M \) can be expressed
  as the finite product of elements in \( N \).
\end{definition}

\begin{definition}
  Let \( M \) be a monoid and let \( X \) and \( Y \) be formal products of
  elements of \( M \). The equation \( X = Y \) is a \textit{relation} for \( M
  \) if the equation holds in \( M \).
\end{definition}

\begin{definition}
  A monoid \( M \) is \textit{commutative} if for all \( a, b \in M \), \( M \)
  is subject to the relation \( ab = ba \).
\end{definition}

\begin{remark}
  For commutative monoids, we use additive notation (by convention). So the above
  relation can be written as \( a + b = b + a \).
\end{remark}

\begin{definition}
  A monoid \( M \) is \textit{idempotent} if for all \( a \in M \), \( M \) is
  subject to the relation \( aa = a \).
\end{definition}

\newcommand{\concat}{\doubleplus}

\begin{definition}
  Let \( S \) be a set. The \textit{free monoid over \( S \)} is the monoid \( M
  \) whose elements are finite formal products of elements of \( S \). The
  product of \( a, b \in M \) is the concatenation \( ab \), and the identity is
  the empty product. To denote the free monoid over \( S \), we use the notation
  \( (\List[S], \concat, []) \).
\end{definition}

\begin{example}
  Let \( S \) be a set. Then the powerset of \( S \), denoted \( \Set{S} \), is a
  commutative, idempotent monoid where addition is set union and the identity is
  the empty set \( \emptyset \in \Set{S} \).
\end{example}

\begin{definition}
  Let \( (M, \mul_M, \id_M) \) and \( (N, \mul_N, \id_N) \) be monoids. A
  function \( f: N \to M \) is a \textit{monoid homomorphism} (or simply
  \textit{homomorphism}) if it preserves identity and monoidal product; that is:

  \begin{itemize}[label={}]
    \item[(a)] \( f(\id_M) = \id_N \).
    \item[(b)] For all \( a, b \in M \), we have \( f(a \mul_M b) = f(a) \mul_N f(b) \).
  \end{itemize}
\end{definition}

\begin{example}
  Let \( f : \List{S} \to \Set{S} \) be the function that sends a formal product
  to its set; i.e.
  \[
    f(a_1a_2 \ldots a_k) = \{ a_1, a_2, \ldots, a_k \}.
  \]
  It is easy to check that \( f \) preserves identity and products and thus is a
  homomorphism.
\end{example}

\begin{propositionapxrep}
  The composition of homomorphisms is a homomorphism.
\end{propositionapxrep}

\begin{proof}
  The proof is standard. Let \( f \) and \( g \) be homomorphisms. We check
  \begin{align*}
    (f \compose g)(\id) = f(g(\id)) = f(\id) = \id
  \end{align*}
  and
  \begin{align*}
    (f \compose g)(ab) = f(g(ab)) = f(g(a) \mul g(b))
    = f(g(a)) \mul f(g(b)) = (f \compose g)(a) \mul (f \compose g)(b)
  \end{align*}
\end{proof}

\subsection{Stream functions}
\label{sec:stream-functions}

\newcommand{\Update}[1]{\Delta#1}

\begin{definition}[Stream function]
  \label[definition]{def:stream-function}
  Let \( M \) and \( N \) be monoids. A function \( F: M \to N \)
  is a \textit{stream function} if there exists a function
  \( \Update{F}: M \times M \to N \) such that for all \( p, a, b \in M \) we have:
  \begin{align}
    F(pa)              & = F(p) \mul \Update F (p, a)
    \tag{1} \label{eqn:derivative}
    \\
    \Update F (p, \id) & = \id
    \tag{2a} \label{eqn:regular-id}
    \\
    \Update F (p, ab)  & = \Update F (p, a) \mul \Update F (pa, b)
    \tag{2b} \label{eqn:regular-mul}
  \end{align}
\end{definition}

\noindent
How should we interpret the above definition? Stream functions are computations
whose outputs can be updated appropriately when new input arrives. That is,
suppose \( P: M \to N \) is a stream function. Let \( p, a \in M \) and suppose
\( p \) arrives first, then \( a \) arrives. The computation \( P(pa) \)
represents the output of the stream function if it processes both \( p \) and
\( a \) all in one batch. Otherwise, suppose \( P \) first processes \( p \)
with output \( P(p) \). When input \( a \) arrives, \( P \) emits \(
\Update{P}(p, a) \) to update its output. Condition~(1) requires that \( P(pa)
= P(p) \mul \Update{P}(p, a) \), which ensures the result is the same.

Condition~(2)\footnotemark{} says that \( \Update{F} \) respects the monoidal
structure for \( M \)---that is, \( \Update{F} \) computes the same resulting
update regardless of the factorization of arriving inputs. We need it for a
step in the (constructive) proof of \cref{thm:decomposition}, which constructs
a stream processor from a stream function.

\footnotetext{
  \Citet[p.~4, defn.~2.5]{DBLP:conf/fossacs/Alvarez-Picallo19} call \crefrange{eqn:regular-id}{eqn:regular-mul} \emph{regularity}.
}

Morally, condition~(1) is the only condition we need for a function to be a
stream function. The next two (constructive) propositions show that for many
monoidal structures, if \( F : M \to N \) is a function with an update function
\( \Update{F} \) satisfying condition~(1), then \( F \) has an update function
satisfying condition~(2); i.e. \( F \) is a stream function. These structures
include the common monoidal structures used in stream programs: lists, sets,
bags, \( \mathbb{Z} \)-sets, and (semi)-lattices. Also, if we define \(
\Update{F} \) by the generators of \( M \) for the second argument, then
condition~(2) is automatically satisfied.

\begin{propositionapxrep}
  \label{prop:left-cancel-stream-function}
  Suppose \( F : M \to N \) is a function equipped with \( \Update{F} \) that
  satisifes condition~(1) of \cref{def:stream-function}, and \( \Img{F} \) is a
  left-cancellative monoid. Then \( F \) is a stream function.
\end{propositionapxrep}

\begin{proof}
  \label{apxpf:stream-function-props}
  We check \( \Update{F} \) satisfies condition~(2). For condition~(2a):
  \
  \begin{align*}
    F(p) \mul \Update{F}(p, \id) & = F(p\id) \\
                                 & = F(p)    \\
    \implies \Update{F}(p, \id)  & = \id
  \end{align*}
  \
  For condition~(2b):
  \begin{align*}
    F(p) \mul \Update{F}(p, ab) & = F(pab)                                            \\
                                & = F(pa) \mul \Update{F}(pa, b)                      \\
                                & = F(p) \mul \Update{F}(p, a) \mul \Update{F}(pa, b) \\
    \implies \Update{F}(p, ab)  & = \Update{F}(p, a) \mul \Update{F}(pa, b).
  \end{align*}
\end{proof}

\begin{propositionapxrep}
  \label{prop:idempontent-stream-function}
  Suppose \( F : M \to N \) is a function equipped with \( \Update{F} \) that
  satisifes condition~(1) of \cref{def:stream-function}, and \( \Img{F} \) is
  an idempotent monoid. Then \( F \) is a stream function.
\end{propositionapxrep}

\begin{proof}
  Define \( \Update{F}' : M \times M \to N \) by \
  \[
    \Update{F}'(p, a) = \begin{cases}
      \id                        & \text{if \( a = \id \)} \\
      F(p) \mul \Update{F}(p, a) & \text{otherwise}.       \\
    \end{cases}
  \]
  We check \( \Update{F}' \) satisfies conditions (1) and (2). Condition~(1): \
  \begin{align*}
    F(pa) & = F(p) \mul \Update{F}(p, a) \tag{condition~(1) on \( \Update{F} \)} \\
          & = F(p) \mul F(p) \mul \Update{F}(p, a)  \tag{idempotence}            \\
          & = F(p) \mul \Update{F}'(p, a)                                        \\
  \end{align*}
  \
  Clearly, \( \Update{F}' \) satisfies condition~2a by definition. For
  condition~2b, there are four cases; it is easy to check the cases
  where \( a = \id \) or \( b = \id \). Otherwise, we have
  \
  \begin{align*}
    \Update{F}'(p, ab) & = F(p) \mul \Update{F}(p, ab)                                  \\
                       & = F(pab)
    \tag{condition~(1) on \( \Update{F} \)}                                             \\
                       & = F(pa) \mul \Update{F}(pa, b)
    \tag{condition~(1) on \( \Update{F} \)}                                             \\
                       & = F(pa) \mul F(pa) \mul \Update{F}(pa, b)
    \tag{idempotence}                                                                   \\
                       & = F(p) \mul \Update{F}(p, a) \mul F(pa) \mul \Update{F}(pa, b)
    \tag{condition~(1) on \( \Update{F} \)}                                             \\
                       & = \Update{F}'(p, a) \mul \Update{F}'(pa, b)
  \end{align*}
\end{proof}

\noindent
Full proofs of the above two propositions can be found in the
\hyperref[apxpf:stream-function-props]{Appendix}.

\begin{example}
  If \( f: M \to N \) is a monoid homomorphism, then it is a stream function: let
  \( \Update{f}(p, b) = f(b) \). Intuitively, \( \Update{f} \) is stateless, as
  it ignores the previous input \( p \). Observe that \( f(ab) = f(a) \mul f(b) =
  f(a) \mul \Update{f}(a, b) \).
\end{example}

\begin{example}
  \label[example]{ex:prefix-sum}
  Prefix sum is a stream function, but not a monoid homomorphism.
  Let \( (\List[\Int], \concat, []) \) be the free monoid over the integers (i.e.\ lists of integers), and \( \prefixSum : \List[\Int] \to \List[\Int] \) map a list to its prefix sum.
  For example, \( \prefixSum([1, 2, 3]) = [1, 3, 6] \).
  This is not a monoid homomorphism, since
  \[
    \prefixSum([1, 2, 3])
    = [1,3,6]
    \neq [1] \concat [2, 5]
    = \prefixSum([1]) \concat \prefixSum([2,3]).
  \]
  However, it \emph{is} a stream function with \( \Update{\prefixSum}(xs, ys) =
  [\left(\sum xs\right) + y : y \in \prefixSum(ys)] \). For instance, \(
  \Update{\prefixSum}([1], [2, 3]) = [1 + 2, 1 + 5] = [3, 6] \). Any program that
  implements \( \prefixSum \) in streaming fashion must be \emph{stateful,} since
  it must track the cumulative sum. This is why \( \prefixSum([1, 2, 3]) \neq
  \prefixSum([1]) \concat \prefixSum([2, 3]) \): the right hand side ``restarts''
  \( \prefixSum \) on \( [2, 3] \), forgetting accumulated state. In contrast,
  homomorphisms are \emph{stateless} (as we will make precise in
  \propref{prop:stateless}).
\end{example}

A well-known technique in functional programming is to model stateful
computations as functions which take in a current state and return an updated
state and some output. Under this perspective, we can decompose \( F \) into a
homomorphism into \( \State{\Int}{L} \) followed by an evaluation.

\begin{definition}[State monoid]
  \label[definition]{def:state}
  Given a set \( S \) and monoid \( M \), define the monoid \( \State{S}{M} \) as%
  \footnote{
    The product \( \odot \) is equivalent both to Kleisli composition
    in the Writer monad, and to the lifting of \( {\mul}_M \) into the
    State monad. In Haskell:
    \begin{equation*}
      \alpha \odot \beta
      = \alpha \kleisli_{\mathrm{Writer}\,M} \beta
      = \mathrm{liftM2}_{\mathrm{State}\,S} \;(\cdot_A) \;\alpha \;\beta
    \end{equation*}
    We call this monoid ``State'' to emphasize the stateful nature of stream
    processors, and because we use a construction equivalent to Haskell's
    \( \fmap_{\mathrm{State}} \) in \cref{sec:sequential}.
  }
  \begin{align*}
    \State{S}{M}       & = (S \to S \times M,\, \odot,\, s \mapsto (s, \id_M)) \\
    \alpha \odot \beta & = s \mapsto
    \begin{aligned}[t]
       & \text{let } (s_{\alpha},\, a) = \alpha(s)         \text{ in} \\
       & \text{let } (s_{\beta},\, b)  = \beta(s_{\alpha}) \text{ in} \\
       & (s_\beta,\, a \mul_M b)
    \end{aligned}
  \end{align*}
\end{definition}

\noindent
One can think of an element \( \alpha: S \to S \times M \) of the state monoid
as a ``hypothetical'' over state: given some state \( s \), the element \(
\alpha \) will update the state to \( s_\alpha \) and emit some output \( a \).
The product \( \alpha \odot \beta \) is a hypothetical that runs \( \alpha \)
and \( \beta \) sequentially, and emits the combined output \( a \mul b \).

The result of evaluations of \( \State{S}{M} \) are tuples \( S \times M \).
For clarity, we project out of these tuples using functions \( \st \) and \(
\out \).

\begin{definition}
  Let \( \State{S}{M} \) be the State monoid. Let \( \st : S \times M \to S \) be
  the first projection and \( {\out}: S \times M \to M \) the second projection.
\end{definition}

\noindent
We now show that we can decompose any stream function \( M \to N \) into a
\textit{homomorphism} \( M \to \State{S}{N} \) with initial state and initial
output. We will use the decomposition to define \textit{stream processors} in
\cref{sec:stream-processor}, and we will show that stream processors and stream
functions represent the same class of programs.

\begin{theorem}[Decomposition]
  \label{thm:decomposition}
  Let \( F: M \to N \) be a stream function.
  Then \( F \) can be decomposed into a tuple \( (S, f, s_{\id}, o_{\id}) \)
  where \( f: M \to \State{S}{N} \) is a homomorphism,
  \( s_{\id} \in S \) is an initial state,
  and \( o_{\id} \in N \) is an initial output
  such that for all \( m \in M \), we have \( F(m) = o_{\id} \mul \out{(f(m)(s_{\id}))} \).
\end{theorem}

\noindent
To understand \cref{thm:decomposition}, it's helpful to decompose prefix sum
(\cref{ex:prefix-sum}). Let \( S = \Int \). We define \( f: \List[\Int] \to
\State{\Int}{\List[\Int]} \) first on the generators of \( \List[\Int] \)
(singleton lists) by \( f([n]) = s \mapsto (s+n, [s+n]) \). That is, the input
\( [n] \) corresponds to a function which takes in a cumulative sum, updates
the cumulative sum to \( s+n \), and emits output \( [s+n] \). We extend \( f\,
\) to all of \( \List[\Int] \) by defining \( f(a \concat b) = f(a) \odot f(b)
\). Observe that:
\begin{align*}
  f([1])       & = s \mapsto (s+1, [s+1])      \\
  f([2, 3])    & = s \mapsto (s+5, [s+2, s+5]) \\
  f([1, 2, 3]) & = f([1]) \odot f([2, 3])
  \\ &= s \mapsto (s+6, [s+1, s+3, s+6])
\end{align*}





\noindent
We set the initial state \( s_{\id} = 0 \) and the initial output \( o_{\id} =
  [] \). Then we can verify the final requirement of \cref{thm:decomposition}: \
\begin{align*}
  o_{\id} \concat \out(f([1, 2, 3])(s_{\id}))
  = [] \concat \out(f([1, 2, 3])(0))
  = \out((6, [1, 3, 6]))
  = [1, 3, 6]
\end{align*}
as desired.

\begin{proof}[Proof of \cref{thm:decomposition}]
  Let \( S = M \) and \( f(m) = s \mapsto (s \mul m,\, \Update{F}(s, m)) \).
  Then \( f \) is a homomorphism:
  \
  \begin{align*}
    f(\id) & = s \mapsto (s \id,\, \Update{F}(s, \id))
    \\
           & = s \mapsto (s, \id_M) = \id_{\State{S}{M}}
    \tag{by \cref{eqn:regular-id} from \cref{def:stream-function}}                            \\[6pt]
    f(ab)  & = s \mapsto (sab,\, \Update{F}(s, ab))                                           \\
           & = s \mapsto (sab,\, \Update{F}(s, a) \mul \Update{F}(sa, b))
    \tag{by \cref{eqn:regular-mul} from \cref{def:stream-function}}                           \\
           & = (s \mapsto (sa,\, \Update{F}(s, a))) \odot (s \mapsto (s b, \Update{F}(s, b))) \\
           & = f(a) \odot f(b)
  \end{align*}
  \noindent
  Let the initial state \( s_{\id} = \id_M \) and the initial output \( o_{\id} =
  F(\id_M) \). We verify \( F(m) = o_{\id} \mul \out(f(m)(s_{\id})) \): \
  \begin{align*}
    o_{\id} \mul \out(f(m)(s_{\id}))
     & = o_{\id} \mul \out((m, \Update{F}(\id_M, m))) \\
     & = F(\id_M) \mul \Update{F}(\id_M, m)           \\
     & = F(m)
  \end{align*}
\end{proof}

\noindent
Decomposition is not unique. Above, we decomposed \cref{ex:prefix-sum} with
intermediate state \( S = \Int \), but the ``generic'' construction in the
proof of \cref{thm:decomposition} would create a decomposition with
intermediate state \( S = \List[\Int] \). The functions \(
\State{\Int{}}{\List[\Int]} \) only need to remember the cumulative sum, while
in the generic construction, the functions \( \State{\List[\Int]}{\List[\Int]}
\) encode the entire input history. The first is desirable as \( \Int \) is a
more compact representation of the necessary intermediate state.

\subsection{Stream processors}
\label[section]{sec:stream-processor}

Let \( M \) and \( N \) be monoids and \( S \) be a set. Motivated by
\cref{thm:decomposition}, we define \textit{stream processors}:

\begin{definition}[Stream processor]
  \label[definition]{def:stream-processor}
  A \textit{stream processor} \( \processor{}: M \leadsto N \) is a tuple
  \( (S, f, s_{\id}, o_{\id}) \), where \( f: M \to \State{S}{N} \) is a homomorphism,
  \( s_{\id} \in S \) is an initial state, and \( o_{\id} \in N \) is an initial output.
\end{definition}

\noindent
We claim that stream functions are an \textit{abstract semantics} for stream
programs and stream processors are \textit{syntactical implementations}. That
is:

\begin{enumerate}
  \item Stream functions provide a simple, abstract semantic model, but it is not
        realistic to model real-world stream programs as stream functions: at any time
        (including on update), stream functions have access to the entire input
        history.

  \item Stream processes are a ``syntax''
        for stream programs because (we will show) we can manipulate them symbolically
        using equational reasoning to prove semantic properties. We develop these
        symbolic manipulations in the rest of this section.

  \item Stream processes are implementations because real-world stream programs are
        usually implemented as the composition of many stream processors that emit
        incremental updates with access to (possibly bounded) internal state.
\end{enumerate}

\noindent
It would be great if our semantics (stream functions) and syntax (stream
processors) coincide; i.e, they model the same class of programs. Fortunately,
they do!

\begin{definition}[Running a stream processor]
  Let \( \processor{}: M \leadsto N = (S, f, s_{\id}, o_{\id}) \) be a stream processor.
  Define the function \( \run{} : (M \leadsto N) \to (M \to N) \),
  which takes a stream processor, an input, and produces an output
  \[
    \run{} (S, f, s_{\id}, o_{\id}) = m \mapsto o_{\id} \mul \out(f(m)(s_{\id})).
  \]
\end{definition}

\begin{theorem}[Expressive Soundness]
  \label[theorem]{thm:soundness}
  Let \( \processor{} \) be a stream processor. Then
  \( F = \run \processor{} \) is a stream function. We say that \( F \) is the
  \textnormal{semantics} of \( \processor{} \), and we define
  \( \denote{\processor{}} \defeq \run \processor{} \).
\end{theorem}

\begin{theorem}[Expressive Completeness]
  \label[theorem]{thm:completeness}
  Let \( F: M \to N \) be a stream function. Then there exists a stream
  processor \( \processor{F}: M \leadsto N \) such that
  \( \denote{\processor{F}} = F \).
\end{theorem}

\begin{proof}
  This is really a restatement of \cref{thm:decomposition} on the decomposition
  of stream functions.
\end{proof}

\noindent
Before proving \fullref{thm:soundness}, we introduce some simplifying notation
for the states and outputs of stream processors.

\begin{definition}
  \label[definition]{def:notation}
  Suppose \( \processor{} = (S, f, s_{\id}, o_{\id}) \) is a stream processor.
  Let \( p, q \in M \).
  When \( \processor{} \) is clear from context, we define:
  \begin{itemize}
    \item  \( s_{p} \defeq \st(f(p)(s_{\id})) \)
          to be \( \processor{} \)'s state after processing input \( p \),
    \item  \( o_{p \leadsto q} \defeq \out(f(q)(s_{p})) \)
          to be the additional output \( \processor{} \) emits
          when processing input \( q \) after having already processed input \( p \), and
    \item  \( o_{q} \defeq o_{\id} \mul o_{\id \leadsto q} \)
          to be  the processor's total output after processing input \( q \).
  \end{itemize}
\end{definition}

\noindent
Note that, when \( p = \id \), the output \( o_{\id \leadsto q} \) is \(
\processor{} \)'s \textit{incremental output} after processing input \( q \)
(that is, \( o_{\id \leadsto q} \) does not include the initial output \(
o_{\id} \)).

\begin{lemmaapxrep}
  \label{lem:notation}
  For all \( a, b, p, q \in M \), we have:\
  \begin{enumerate}
    \item \( o_{p \leadsto \id} = \id_N \)
    \item \( o_{p \leadsto ab} = o_{p \leadsto a} \mul o_{pa \leadsto b} \)
    \item For all \( a, b \in M \), we have \( o_{ab} = o_{a} \mul o_{a \leadsto b} \)
    \item For all \( p \in M \), we have \( \denote{\processor{}}(p) = o_{p} \).
  \end{enumerate}
\end{lemmaapxrep}

\noindent
The proof is straightforward and is in the \hyperref[prf:notation]{Appendix.}
Item (1) validates that \( o_{q} \) is well-defined when \( q = \id \) in
\cref{def:notation}. Items (2) and (3) tells us how incremental outputs
decompose. Item (4) relates our new notation to the semantics of \(
\processor{} \).

\begin{proof}
  \label{prf:notation}
  Item (1) is easy to verify by expanding out the definition. Since \( f \)
  is a homomorphism, we have
  \( o_{p \leadsto \id} = \out(f(\id)(s_p)) = \out((s_p, \id)) = \id \).
  To show (2), we check
  \begin{align*}
    o_{p \leadsto ab} & = \out(f(ab)(s_{p}))                        \tag{definition}                \\
                      & = \out((f(a) \odot f(b))(s_{p}))            \tag{\( f \) is a homomorphism} \\
                      & = \out(f(a)(s_{p})) \mul \out(f(b)(s_{pa})) \tag{*}                         \\
                      & = o_{p \leadsto a} \mul o_{pa \leadsto b}.  \tag{definition}
  \end{align*}

  For step (*), eta-expand \( f(a) \) and \( f(b) \) on the line prior, apply the
  definition of \( \odot \), and simplify.

  Item (3) follows from item (2) by setting \( p = \id{} \). Item (4) is
  immediate from the definitions: \
  \[
    \denote{\processor{}}(p) = o_{\id} \mul \out(f(p)(s_{\id})) = o_{\id} \mul o_{\id \leadsto p} = o_{p}.
  \]
\end{proof}

\begin{proof}[Proof of \cref{thm:soundness} (\nameref*{thm:soundness})]
  Suppose \( \processor{F} = (S, f, s_{\id}, o_{\id}) \) is a stream
  processor and let \( F = \denote{\processor{F}} \).
  Define \( \Update{F}(p, a) = o_{p \leadsto a} \). We verify that \( F \) is a
  stream function:
  \
  \begin{align*}
    F(p) \mul \Update{F}(p, a) & = o_{p} \mul o_{p \leadsto a}                   \tag{\cref{lem:notation}} \\
                               & = o_{pa}                                        \tag{\cref{lem:notation}} \\
                               & = F(pa)
  \end{align*}
  and
  \
  \begin{alignat*}{3}
     & \Update{F}(p, \id) &  & = o_{p \leadsto \id} = \id                 \tag{\cref{lem:notation}} \\
     & \Update{F}(p, ab)  &  & = o_{p \leadsto ab}                                                  \\
     &                    &  & = o_{p \leadsto a} \mul o_{pa \leadsto b}  \tag{\cref{lem:notation}} \\
     &                    &  & = \Update{F}(p, a) \mul \Update{F}(pa, b).
  \end{alignat*}
\end{proof}

\newcommand{\unit}{\ast}

\begin{definition}[Pure]
  \label[definition]{def:pure}
  When \( f: M \to N \) is a homomorphism, define \( \pure f \) to be the
  tuple \( \pure f \defeq (\{ \unit{} \}, f', \unit{}, \id_N) \) where
  \( f'(a) = \unit{} \mapsto (\unit{}, f(a)) \).
\end{definition}

\begin{proposition}
  \label[proposition]{prop:pure}
  \( \pure f : M \leadsto N \) is a stream processor with \( \denote{\pure f} = f \).
\end{proposition}

\begin{proof}
  By definition, the initial output of \( \pure f \) is \( \id_N = o_{\id} \).
  Therefore for all \( p \in M \), we have
  \[
    \denote{\pure f}(p) = o_{p} = o_{\id} \mul o_{\id \leadsto p} = o_{\id \leadsto p}.
  \]
  It is easy to check that \( f': M \to \State{\{\unit\}}{N} \) is a
  homomorphism. We conclude
  \[
    \denote{\pure f}(a) = o_{\id \leadsto a} = f(a).
  \]
\end{proof}

\begin{proposition}[Homomorphisms are stateless processors]
  \label[proposition]{prop:stateless}
  \( f: M \to N \) is a homomorphism if and only if there exists
  a stream processor \( \processor{f}: M \leadsto N \) with trivial state
  \( S = \{\unit{}\} \) and trivial initial output \( o_{\id} = \id_N \) where
  \( \denote{\processor{f}} = f \).
\end{proposition}

\begin{proof}
  The forward direction is implied by \propref{prop:pure}.
  Conversely, suppose \( \processor{f}: M \leadsto N \) is a stream processor \(
  \processor{f} = (S, f', \unit{}, \id_N) \). Let \( f = \denote{\processor{f}}
  \) We check \( f \) is a homomorphism: \
  \begin{alignat*}{3}
    f(\id_M)              & = o_{\id} = \id_N            &                                                                                                                          \\[3px]
    f(ab)                 & = \denote{\processor{f}}(ab)
    = o_{\id \leadsto ab} &                              & = o_{\id \leadsto a} \mul o_{a \leadsto b}
    \tag{\lemref{lem:notation}}                                                                                                                                                     \\
                          &                              &                                            & = o_{\id \leadsto a} \mul o_{\id \leadsto b} \tag{*}                        \\
                          &                              &                                            & = \denote{\processor{f}}(a) \mul \denote{\processor{f}}(b) = f(a) \mul f(b)
  \end{alignat*}
  where step (*) follows since
  \
  \[
    o_{a \leadsto b}
    = \out(f(b)(s_{a}))
    = \out(f(b)(\unit{}))
    = \out(f(b)(s_{\id{}}))
    = o_{\id \leadsto b}.
  \]
  Observe \( \{\unit{}\} \) has exactly one value so \( s_{a} = \unit{} = s_{\id}
  \). (This reasoning is reminiscent of ``The Trick'' from partial evaluation; we
  will say more in \cref{sec:thetrick}.) We conclude \( f \) is a homomorphism.
\end{proof}

\subsection{Sequential composition}
\label{sec:sequential}

\begin{definition}
  \label[definition]{def:sequential}
  Suppose \( \processor{}: M \leadsto N \) and
  \( \processortwo{}: N \leadsto P \) are stream processors with
  \( \processor{} = (S, f, s_{\id}, n_{\id}) \) and
  \( \processortwo{} = (T, g, t_{\id}, p_{\id}) \).
  The composition \( \processor{} \pipe{} \processortwo{} \) is a stream
  processor defined by
  \begin{align*}
    \processor{} \pipe{} \processortwo{}          & : M \leadsto P                                                                        \\
    \processor{} \pipe{} \processortwo{}          & \defeq (S \times T, h, (s_{\id}, t_{n_{\id}}), p_{\id} \mul p_{\id \leadsto n_{\id}}) \\
    \text{where}                                                                                                                          \\
    h(m)                                          & = (s, t) \mapsto
    \begin{aligned}[t]
       & \text{let } (s', n) = f(m)(s) \text{ in} \\
       & \text{let } (t', p) = g(n)(t) \text{ in} \\
       & ((s', t'), p)
    \end{aligned}                                                                                           \\
    (t_{n_{\id{}}}, p_{\id{} \leadsto n_{\id{}}}) & = g(n_{\id{}})(t_{\id{}}) \tag{\Cref{def:notation}}
  \end{align*}

\end{definition}

\noindent
The composition of two stream processors \( \processor{} \) and \(
\processortwo{} \) is a stream processor that runs both processors in sequence,
keeping internal state \( S \times T \). When given input \( m \), it feeds \(
m \) into \( \processor{} \) to obtain output \( n \) and updated state \( s'
\) for \( \processor{} \). It then feeds \( n \) into \( \processortwo{} \) to
obtain output \( p \) and updated state \( t' \) for \( \processortwo{} \). To
obtain the composed processor's initial state and output, we need to send \(
\processor{} \)'s initial output into \( \processortwo{} \); the resulting
initial state is \( t_{n_{\id{}}} \) and initial output is \( p_{\id{} \leadsto
    n_{\id{}}} \).

We connect the syntax \( \processor{} \pipe{} \processortwo{} \) to the
semantics in the following proposition. It states that the semantics of a
sequential composition of processors is the sequential composition of their
semantics.

\begin{propositionapxrep}[\( \denote{-} \) is a functor]
  \label{prop:functor}
  Let \( \processor{} \) and \( \processortwo{} \) be stream processors.
  Then
  \begin{align*}
    \denote{\pure \idfn} & = \idfn{}                                              \tag{identity}    \\
    \denote{\processor{} \pipe \processortwo{}}
                         & = \denote{\processor{}} \pipe \denote{\processortwo{}} \tag{composition}
  \end{align*}
\end{propositionapxrep}

\begin{proofsketch}
  For identity, apply \propref{prop:pure}. For composition, expand out the
  definitions. A full derivation is in the \hyperref[apxpf:functor]{Appendix}.
\end{proofsketch}

\begin{proof}
  \label{apxpf:functor}
  Let \( F = \denote{\processor{}} \) and \( G = \denote{\processortwo{}} \), and
  \( h \) be as defined in \cref{def:sequential}.
  Observe that
  \begin{align*}
    \denote{\processor{} \pipe \processortwo{}} (m)
     & = \run{} (\processor{} \pipe \processortwo{}) (m)                                                                                                                 \\
     & = p_{\id} \mul p_{\id \leadsto n_{\id}} \mul \out(h(m)((s_{\id}, t_{n_{\id}})))                                                                                   \\
     & = p_{\id} \mul p_{\id \leadsto n_{\id}} \mul \out(g(\out(f(m)(s_{\id})))(t_{n_{\id}})) \tag{expanding \( h(m) \)}                                                 \\
     & = p_{\id} \mul p_{\id \leadsto n_{\id}} \mul \out(g(n_{\id \leadsto m})(t_{n_{\id}}))                                                                             \\
     & = p_{\id} \mul p_{\id \leadsto n_{\id}} \mul p_{n_{\id} \leadsto n_{\id \leadsto m}}                                                                              \\
     & = G(\id) \mul \Update{G}(\id, n_{\id}) \mul \Update{G}(n_{\id}, n_{\id \leadsto m})    \tag{\( G(\id) = p_{\id} \) and \( \Update{G}(a, b) = p_{a \leadsto b} \)} \\
     & = G(\id) \mul \Update{G}(\id, F(\id)) \mul \Update{G}(F(\id), \Update{F}(\id, m))      \tag{\( F(\id) = n_{\id} \) and \( \Update{F}(a, b) = n_{a \leadsto b} \)} \\
     & = G(F(\id) \mul \Update{F}(\id, m))                                                    \tag{\( G(ab) = G(a) \mul \Update{G}(a, b) \)}                             \\
     & = G(F(m))                                                                                                                                                         \\
     & = (F \pipe G)(m) = (\denote{\processor{}} \pipe \denote{\processortwo{}})(m)
  \end{align*}
\end{proof}

\begin{definition}
  We say stream processors \( \processor{} \) and \( \processortwo{} \) are
  equivalent, written as \( \processor{} \sim \processortwo{} \), if \(
  \denote{\processor{}} = \denote{\processortwo{}} \).
\end{definition}

\begin{example}
  In the proof of \cref{thm:decomposition}, we decomposed prefix sum
  (\cref{ex:prefix-sum}) in two different ways. The two decompositions induce
  non-equal but equivalent stream processors.
\end{example}

\begin{corollary}
  The composition of stream processors is associative, up to equivalence.
  That is, for all stream processors \( \processor{} \), \( \processortwo{} \),
  and \( \processorthree{} \), we have
  \
  \[
    (\processor{} \pipe \processortwo{}) \pipe \processorthree{}
    \sim \processor{} \pipe (\processortwo{} \pipe \processorthree{}).
  \]
\end{corollary}

\begin{proof}
  By applications of \propref{prop:functor} and associativity of function
  composition.
\end{proof}

\begin{propositionapxrep}[Fusion]
  \label{prop:fusion}
  Suppose \( f: M \to N \) is a homomorphism and
  \( \processor{}: N \leadsto P = (S, g, s_{\id}, o_{\id}) \) is a stream
  processor. Let
  \( \processortwo{} = (S, f \pipe g, s_{\id}, o_{\id}) \). Then
  \( \processortwo{}: M \to P \) is a stream processor and
  satisfies
  \[
    \processortwo{} \sim \pure f \pipe \processor{}.
  \]
\end{propositionapxrep}

\begin{proofsketch}
  Expand out the definitions. A full derivation is in the
  \hyperref[apxpf:fusion]{Appendix}.
\end{proofsketch}

\Propref{prop:fusion} says that if we have a pure processor \( \pure f \)
followed by a stateful operator \( \processor{} \), we can fuse \( f \) with
\( \processor{} \) to create a single operator \( \processortwo{} \) which is
semantically equivalent.

\begin{proof}
  \label{apxpf:fusion}
  Observe \( f \pipe g \) is a homomorphism as it is a composition of
  homomorphisms. So \( \processortwo{} \) is a stream processor. Then we need
  to show that
  \( \denote{\pure{f} \pipe \processor{}} = \denote{\processortwo{}} \):
  \
  \begin{align*}
    \denote{\pure{f} \pipe \processor{}}(m)
     & = (\denote{\pure{f}} \pipe \denote{\processor{}})(m) \tag{\propref{prop:functor}} \\
     & = \denote{\processor{}} (\denote{\pure{f}}(m))                                    \\
     & = \denote{\processor{}} (f(m)) \tag{\propref{prop:pure}}                          \\
     & = o_{\id} \mul \out(g(f(m))(s_{\id}))                                             \\
     & = o_{\id} \mul \out((f \pipe g)(m)(s_{\id}))                                      \\
     & = \denote{\processortwo{}}(m)
  \end{align*}
\end{proof}

\begin{corollary}
  \label[corollary]{cor:decoupling}
  For all homomorphisms \( f \) and \( g \), we have
  \[
    \pure{} (f \pipe g) \sim \pure f \pipe \pure g.
  \]
\end{corollary}

\noindent
\Cref{cor:decoupling} models \textit{pipeline parallelism} of stateless stream
processors. That is, \( \pure{} (f \pipe g) \) represents a single processor
that computes \( f \pipe g \) statelessly on input elements while
\( \pure f \pipe \pure g \) represents a sequence of two processors where the first
processor computes \( f \) on its input, and then sends its output to the second
processor, which computes \( g \).

\begin{definition}
  \label[definition]{def:eval}
  Suppose \( \processor{}: M \leadsto N \) is a stream processor where
  \( \processor{} = (S, f, s_{\id}, o_{\id}) \). Define the processor
  \( \eval \processor{}: \State{S}{N} \leadsto {N} \) as the tuple
  \( \eval \processor{} \defeq (S, \idfn, s_{\id}, o_{\id}) \).
\end{definition}

\begin{corollary}
  \label[corollary]{cor:decompose}
  For all processors \( \processor{} = (S, f, s_{\id{}}, o_{\id{}}) \), we have
  \[
    \processor{} \sim \pure f \pipe \eval{} \processor{}.
  \]
\end{corollary}

\begin{proof}
  Apply \propref{prop:fusion}.
\end{proof}

\begin{remark}
  Observe the type of \( \pure{} f : M \leadsto \State{S}{N} \) in
  \cref{cor:decompose}.
\end{remark}

\noindent
\Cref{cor:decompose} is a syntactic analogy to
\cref{thm:decomposition}. It says that all stateful stream processors
\( \processor{} = (S, f, s_{\id}, o_{\id}) \) can be decomposed into a parallel
portion \( \pure f \) followed by a reduction \( \eval{} \processor{} \).

\begin{definition}
  Let \( f: A \to B \). Define \( f_*: \State{S}{A} \to \State{S}{B} \) by
  \[
    f_*(\alpha) \defeq s \mapsto \text {let } (s', a) = \alpha(s) \text { in } (s', f(a))
  \]
\end{definition}

\begin{remark}
  In Haskell, \( f_* \) is exactly \( \fmap_{\mathrm{State}} f \).
\end{remark}

\noindent
Recall that we can interpret an element \( \alpha \) of \( \State{S}{A} \) as
hypotheticals which, when given a state, updates the state and emits an output
\( a \). For any function \( f: A \to B \), and any \( \alpha \in \State{S}{A}
\), we can imagine ``pushing'' \( f \) into \( \alpha \) to create a function
\( \State{S}{B} \), which is the same as \( \alpha \) except emits \( f(a) \)
instead. The function which ``pushes in'' \( f \) is exactly \( f_* \).

\begin{proposition}
  \( f: M \to N \) is a homomorphism if and only if \( f_* \) is a homomorphism.
\end{proposition}

\begin{proofsketch}
  Consider \( f_*(\alpha \odot \beta) \). By definition of \( \odot \) and \( f_*
  \), we have \
  \begin{align*}
    f_*(\alpha \odot \beta) & = s \mapsto
    \begin{aligned}[t]
       & \text{let } (s_\alpha, a) = \alpha(s) \text{ in}      \\
       & \text{let } (s_\beta, b) = \beta(s_\alpha) \text{ in} \\
       & (s_\beta, f(ab))
    \end{aligned}
  \end{align*}
  Apply \( f(ab) = f(a)f(b) \) to the above and observe that the result is
  equal to \( f_*(\alpha) \odot f_*(\beta) \).

  For the other direction, we need to show that for all \( a, b \in M \), we have
  \( f(ab) = f(a)(b) \). For any \( a \) and \( b \), choose \( \alpha = s
  \mapsto (s, a) \) and \( \beta = s \mapsto (s, b) \), and expand as above.
\end{proofsketch}

\begin{theoremapxrep}[Exchange]
  \label{thm:exchange}
  Let \( \processor{} : \State{S}{M} \leadsto M = (S, \idfn{}, s_{\id{}}, o_{\id{}}) \) and
  \( g : M \to N \) be a homomorphism. Define
  \( \processortwo{} : \State{S}{N} \leadsto N = (S, \idfn{}, s_{\id{}}, g(o_{\id{}})) \).
  Then we have
  \
  \[
    \processor{} \pipe{} \pure{} g \sim \pure{} g_{*} \pipe{} \processortwo{}.
  \]
\end{theoremapxrep}

\noindent
The proof is mechanical and is in the \hyperref[prf:exchange]{Appendix.}

\begin{proof}
  \label{prf:exchange}
  Let \( F = \denote{\processor{}} \). Observe
  \begin{align*}
    \denote{\processor{} \pipe{} \pure{} g}(\alpha)
     & = (F \pipe g)(\alpha)                                    \\
     & = g(F(\alpha))                                           \\
     & = g(F(\id) \mul \Update{F}(\id{}, \alpha))               \\
     & = g(F(\id)) \mul g(\Update{F}(\id{}, \alpha))            \\
     & = g(o_{\id{}}) \mul g(\out(\idfn(\alpha)(s_{\id{}})))    \\
     & = g(o_{\id{}}) \mul \out(\idfn(g_*(\alpha))(s_{\id{}}))) \\
     & = \denote{\pure{} g_{*} \pipe \processortwo{}}(\alpha)
  \end{align*}
\end{proof}

\begin{corollary}
  \label[corollary]{cor:exchangeeval}
  For all processors
  \( \processor{} : M \leadsto N = (S, f, s_{\id{}}, o_{\id{}}) \)
  with identity initial output \( o_{\id} = \id_N \) and
  homomorphisms \( g : N \to P \), we have
  \[
    \eval{} \processor{} \pipe{} \pure{} g \sim \pure{} g_{*} \pipe{} \eval_{g_*} \processor{}
  \]
  where \( \eval_{g_*} \processor{} : \State{S}{P} \leadsto P \defeq (S, \idfn,
  s_{\id}, \id_{P}) \).
\end{corollary}

\begin{proof}
  Apply \cref{thm:exchange} and observe that \( g(o_\id) = g(\id_N) = \id_P \).
\end{proof}

\noindent
The combination of \cref{cor:decompose}, \cref{cor:exchangeeval}, and
\cref{cor:decoupling} are powerful because they let us study the sequential
composition of stateful processors in terms of the composition of their induced
homomorphisms. We show this by example in the next section.

\subsection{Application: the discrete integral and derivative are inverses}
\label{sec:inverse}

To demonstrate the utility of the tools developed in the previous sections, we
take a short detour from specification and prove that two common stream
functions, the discrete integral and the discrete derivative, are inverses
\citep{DBLP:journals/pvldb/BudiuCMRT23}.

We give two proofs: the first is a proof which appeals to the semantics of the
program by analyzing how each stream function transforms its input. The second
proof shows how syntactic transformations allow us to reason about the
program's behavior by analyzing the composition of the processors'
homomorphisms.

\begin{definition}
  Let \( G \) be a group. Define the \textit{discrete integral} \( I: \List[G]
  \to \List[G] \) and the \textit{discrete derivative} \( D: \List[G] \to
  \List[G] \) by \
  \begin{align*}
    I([a_1,\, \ldots,\, a_n])
     & = [a_1,\, a_1 a_2,\, a_1 a_2 a_3,\, \ldots,\, \prod_{i=1}^{n} a_i]
    \\
    D([a_1,\, a_2,\, \ldots,\, a_n])
     & = [\id^{-1} a_1,\, a_1^{-1} a_2 ,\, \ldots,\, a_{n-1}^{-1} a_n]
  \end{align*}
\end{definition}

\begin{theorem}
  \label[theorem]{thm:int-der-inverses}
  The discrete integral and discrete derivative are inverses:
  \( (I \pipe D) = \idfn_{\List[G]} = (D \pipe I) \).
\end{theorem}

\noindent
We give two proofs of the first equality \( (I \pipe D) = \idfn_{\List[G]} \).
A proof of the second equality is symmetrical.

\begin{proof}[First proof of \cref{thm:int-der-inverses}]
  Fix some input \( L = [a_1, a_2, \ldots, a_n] \). For a list \( L \), let \( L_i \)
  denote its \( i \)\textsuperscript{th} entry (so e.g. \( L_1 = a_1 \)). It is
  easy to show by induction that for all \( 1 \leq k \leq n \), we have
  \( I(L)_k = I(L)_{k-1} \mul a_k \). Then, for all \( 1 \leq k \leq n \):
  \
  \[
    (I \pipe D)(L)_k
    = D(I(L))_k
    = I(L)_{k-1}^{-1} \mul I(L)_k
    = \left( I(L)_{k-1}^{-1} \mul I(L)_{k-1} \right) \mul a_k
    = a_k.
  \]
\end{proof}

The second proof will show that the sequential composition of stream processors
that implement \( D \) and \( I \) is the identity processor. Again, one can
think of stream functions as a semantics and stream processors as a syntax. To
prove that the composition of stream functions satisfy some semantic property,
\fullref{thm:soundness} says that we can take syntactic objects that implement
those semantics, manipulate them symbolically in ways that preserve those
semantics, and show that the resulting object is the identity stream processor.

Therefore, we construct (stateful) stream processors \( \processor{} \) and \(
\processortwo{} \) such that \( \denote{\processor{}} = I \) and \(
\denote{\processortwo{}} = D \). The composition \( \processor{} \pipe
\processortwo{} \) is a syntactic object; we manipulate the \( \processor{}
\pipe \processortwo{} \) with the semantics-preserving transformations
developed in the previous section to show that it is equivalent to the identity
processor \( \pure{} \idfn{} \). We conclude then that \( (I \pipe D) =
\denote{\processor{} \pipe \processortwo{}} = \denote{\pure{} \idfn} = \idfn \)
by \propref{prop:functor}.

\begin{proof}[Second proof of \cref{thm:int-der-inverses}]
  Let \( \processor{} = (G, f, \id{}, []) \)
  where \( f \) is defined on the generators by
  \( f([a]) = s \mapsto (s \mul a, [s \mul a]) \).
  Similarly, let \( \processortwo{} = (G, g, \id{}, []) \).
  where \( g([b]) =  t \mapsto (b, [t^{-1} \mul b]) \).
  In \( f \), the state \( s \) is the accumulated prefix sum, and in
  \( g \), the state \( t \) is the value of the previous list element.
  We leave it to the reader to check \( \denote{\processor{}} = I \) and
  \( \denote{\processortwo{}} = D \).
  Consider the composition \( \processor{} \pipe \processortwo{} \):
  \
  \begin{align*}
    \processor{} \pipe \processortwo{}
     & \sim \pure{} f \pipe \eval{} \processor{} \pipe \pure{} g \pipe \eval{} \processortwo{}     \tag{\cref{cor:decompose}}        \\
     & \sim \pure{} f \pipe \pure{} g_* \pipe{} \eval_{g_*} \processor{} \pipe \eval{} \processortwo{} \tag{\cref{cor:exchangeeval}} \\
     & \sim \pure{} (f \pipe g_*) \pipe{} \eval_{g_*} \processor{} \pipe \eval{} \processortwo{}    \tag{\cref{cor:decoupling}}
  \end{align*}
  To understand \( f \pipe g_* \), it is straightforward to check that
  \[
    (f \pipe g_*)([c]) = s \mapsto (s \mul c, (t \mapsto (s \mul c, t^{-1} \mul s \mul c))).
  \]
  Using induction, we can simplify \( \odots_{i=1}^{n} \left( f \pipe g_* \right)
  ([a_i]) \) to show
  \begin{align}
    \label[equation]{eq:gcomposef}
    (f \pipe g_*)([a_1, a_2, \ldots, a_n])
         & = s \mapsto (A(s), (t \mapsto (A(s), [t^{-1} \mul s \mul a_1, a_2, \ldots, a_n]))) \\
    \text{where }
    A(s) & = s \mul \prod_{i=1}^{n} a_i \nonumber
  \end{align}

  \noindent
  (We need to use induction to simplify the product to
  \cref{eq:gcomposef} for arbitrary \( n \). For a fixed \( n \), we do not need
  induction, but that would be equivalent to symbolically evaluating the
  program.)

  That is, for any input \( [a_1, a_2, \ldots, a_n] \), the output of the
  processor \( \pure{} (f \pipe g_*) \) is the function \cref{eq:gcomposef}. This
  function is ``missing'' the initial states \( s \) and \( t \), which \(
  \eval{} \processor{} \pipe \eval{} \processortwo{} \) will provide. When
  evaluated, it will output a list whose elements are necessarily equal to the
  input list for all elements except the first. The value of the first element
  depends on the initial states of \( \processor{} \) and \( \processortwo{} \).

  Observe then \( \eval_{g_*} \processor{} \) will pass initial state \( \id{} \)
  for \( s \), and then \( \eval{} \processortwo{} \) will pass initial state \(
  \id{} \) for \( t \). So at the end of the composed pipeline, the output will
  be
  \[
    [] \concat \out([] \concat \out((f \pipe g_*)([a_1, a_2, \ldots, a_n])(\id{}))(\id{})) = [\id{}^{-1} \mul \id{} \mul a_1, a_2, \ldots, a_n] = [a_1, a_2, \ldots, a_n]
  \]
  Hence
  \[
    (\processor{} \pipe \processortwo{})
    \sim^* (\pure{} (f \pipe g_*) \pipe{} \eval_{g_*} \processor{} \pipe \eval{} \processortwo{})
    \sim (\pure{} \idfn).
  \]
  We conclude \( (I \pipe D) = \denote{\processor{} \pipe \processortwo{}} =
  \denote{\pure{} \idfn} = \idfn \).
\end{proof}

\noindent
In the above proof, we chose stream processors \( \processor{} \) and \(
\processortwo{} \) so that the composition \( (f \pipe g_*) \) ``cancelled
out'' syntactically on all list elements except the first (whose value
necessarily depends on initial states). \fullref{thm:soundness} tells us that
if we want to prove a semantic property through syntactical transformations, it
is correct to choose \textit{any} syntactic objects that implement the
semantics, but proofs will of course be easier if we choose objects
appropriately.

Furthermore, the two proofs of \cref{thm:int-der-inverses} make exactly the
same argument. Specifically, the accumulated state for \( \processor{} \) and
\( \processortwo{} \) are both equal in the simplified product of \(
\odots_{i=1}^{n} \left( f \pipe g_* \right) ([a_i]) \) (\cref{eq:gcomposef}),
which is exactly the inductive hypothesis in the first proof. In this case,
choosing appropriate \( \processor{} \) and \( \processortwo{} \) is equivalent
to choosing an appropriate inductive hypothesis.

\subsection{Parallel composition}
\label{sec:parallel}

\begin{definition}
  \label[definition]{def:direct-product}
  Let \( M \) and \( N \) be monoids. We define the direct product of \( M \) and
  \( N \), written \( M \dprod N \),
  as formal products generated by elements of the form \( (m, \id_N) \) for all
  \( m \in M \) and \( (\id_M, n) \) for all \( n \in N \), subject to the
  relation
  \[
    (m_1, n_1) \mul_{M \dprod N} (m_2, n_2) = (m_1 \mul_M m_2, n_1 \mul_N n_2).
  \]

  That is, the product \( \mul_{M \dprod N} \) is defined componentwise.
\end{definition}

\begin{remark}
  Observe that even when \( M \) and \( N \) are not themselves commutative
  monoids, the elements \( (m, \id) \) and \( (\id, n) \) commute, as
  \[
    (m, \id) \mul (\id, n) = (m \mul \id, \id \mul n) = (\id \mul m, n \mul \id) = (\id, n) \mul (m, \id).
  \]
\end{remark}

Hence, the direct product is a good model for two parallel streams \( M \) and
\( N \).

\begin{definition}
  Let \( f : M \to P \) and \( g : N \to Q \) be functions. We define the direct
  product of \( f \) and \( g \) to be the function \( f \dprod g : M \dprod N
  \to P \dprod Q \) given by: \
  \begin{align*}
    (f \dprod g)((m, n)) & \defeq (f(m), g(n))
  \end{align*}
\end{definition}

The following three propositions verify products \( f \dprod g \) behave as we
expect for homomorphisms and stream functions. Their proofs are straightforward
and are in the \hyperref[prf:product-well-defined]{Appendix.}

\begin{propositionapxrep}
  \( f \dprod g \) is well-defined; that is, it respects relations:
  \[
    (f \dprod g)((m_1 m_2, n_1 n_2)) = (f \dprod g)((m_1, n_1)) \mul (f \dprod g)((m_2, n_2)).
  \]
\end{propositionapxrep}

\begin{proof}
  \label[proof]{prf:product-well-defined}
  \begin{align*}
    (f \dprod g)((m_1 m_2, n_1 n_2))
     & = (f(m_1 m_2), g(n_1 n_2))                               \\
     & = (f(m_1) \mul f(m_2), g(n_1) \mul g(n_2))               \\
     & = (f(m_1), g(n_1)) \mul (f(m_2), g(n_2))                 \\
     & = (f \dprod g)((m_1, n_1)) \mul (f \dprod g)((m_2, n_2))
  \end{align*}
\end{proof}

\begin{propositionapxrep}
  If \( f \) and \( g \) are homomorphisms, then \( f \dprod g \) is a
  homomorphism.
\end{propositionapxrep}

\begin{proof}
  We check \( (f \dprod g)(a_1 a_2) = (f \dprod g)(a_1) \mul (f \dprod g)(a_2)
  \):
  \begin{align*}
    (f \dprod g)((m_1 m_2, n_1 n_2))
     & = (f(m_1 m_2), g(n_1 n_2))                               \\
     & = (f(m_1) f(m_2), g(n_1) g(n_2))                         \\
     & = (f(m_1), g(n_1)) \mul (f(m_2), g(m_2))                 \\
     & = (f \dprod g)((m_1, n_1)) \mul (f \dprod g)((m_2, n_2))
  \end{align*}
\end{proof}

\begin{propositionapxrep}
  If \( F \) and \( G \) are stream functions, then \( F \dprod G \) is a stream
  function.
\end{propositionapxrep}

\begin{proof}
  Since \( F : M \to P \) and \( G : N \to Q \) are stream functions, by
  \cref{def:stream-function}, there exist functions \( \Update{F} : M \dprod M
  \to P \) and \( \Update{G} : N \dprod N \to Q \) satisfying conditions \
  \begin{alignat*}{3}
            & F (ma)     &  & = F(m) \mul \Update{F}(m, a)              \\
    \Update & F (m, \id) &  & = \id                                     \\
    \Update & F (m, ab)  &  & = \Update{F}(m, a) \mul \Update{F}(ma, b) \\
            & G (nc)     &  & = G(n) \mul \Update{G}(n, c)              \\
    \Update & G (n, \id) &  & = \id                                     \\
    \Update & G (n, cd)  &  & = \Update{G}(n, c) \mul \Update{G}(nc, d)
  \end{alignat*}

  Define \( \Update{(F \dprod G)} : (M \dprod N) \times (M \dprod N) \to (P
  \dprod Q) \) by
  \begin{align*}
    \Update{(F \dprod G)}((m, n), (a, c)) = (\Update{F}(m, a), \Update{G}(n, c))
  \end{align*}

  We check conditions (1) and (2). First, condition~(1): \
  \begin{align*}
    (F \dprod G)((m, n) \times (a, c))
     & = (F \dprod G)((ma, nc))                                          \\
     & = (F(ma), G(nc))                                                  \\
     & = (F(m) \mul \Update{F}(m, a), G(n) \mul \Update{G}(n, c))        \\
     & = (F(m), G(n)) \mul (\Update{F}(m, a), \Update{G}(n, c))          \\
     & = (F \dprod G)((m, n)) \mul \Update{(F \dprod G)}((m, a), (n, c))
  \end{align*}

  Condition~(2a):
  \[
    \Update{(F \dprod G)}((m, n), (\id, \id)) = (\Update{F}(m, \id), \Update{G}(n, \id)) = (\id, \id) = \id
  \]

  Condition~(2b):
  \begin{align*}
    \Update{(F \dprod G)}((m, n), (ab, cd))
     & = (\Update{F}(m, ab), \Update{G}(n, cd))                                                \\
     & = (\Update{F}(m, a) \mul \Update{F}(ma, b), \Update{G}(n, c) \mul \Update{G}(nc, d))    \\
     & = (\Update{F}(m, a), \Update{G}(n, c)) \mul (\Update{F}(ma, b), \mul \Update{G}(nc, d)) \\
     & = \Update{(F \dprod G)}((m, n), (a, c)) \mul \Update{(F \dprod G)}((ma, nc), (b, d))
  \end{align*}
\end{proof}

\newcommand{\either}{\oplus}

\begin{definition}
  \label[definition]{def:parallel}
  Suppose \( \processor{}: M \leadsto P \) and
  \( \processortwo{}: N \leadsto Q \) are stream processors with
  \( \processor{} = (S, f, s_{\id}, p_{\id}) \) and
  \( \processortwo{} = (T, g, t_{\id}, q_{\id}) \). We define the product processor
  \( \processor{} \dprod \processortwo{} \) as
  \
  \begin{align*}
    \processor{} \dprod \processortwo{}
              & : M \dprod N \leadsto P \dprod Q                               \\
    \processor{} \dprod \processortwo{}
              & \defeq (S \times T, h, (s_{\id}, t_{\id}), (p_{\id}, q_{\id})) \\
    \text{where}                                                               \\
    h         & : M \times N \to \State{S \times T}{P \times Q}                \\
    h((m, n)) & = (s, t) \mapsto
    \begin{aligned}[t]
       & \text{let } (\alpha, \beta) = (f \dprod g)((m, n)) \text{ in} \\
       & \text{let } (s', p) = \alpha(s)
      \text{ and } (t', q) = \beta(t) \text{ in}                       \\
       & ((s', t'), (p, q))
    \end{aligned}
  \end{align*}
  It is straightforward to verify that \( h \) is a homomorphism.
\end{definition}

\noindent
Again, the following proposition connects our syntax to our semantics in the
expected way.
\begin{propositionapxrep}
  Suppose \( \processor{}: N \leadsto P \) and \( \processortwo{}: M \leadsto Q
  \) are stream processors. We have
  \[
    \denote{\processor{} \dprod \processortwo{}} = \denote{\processor{}} \dprod \denote{\processortwo{}}
  \]
\end{propositionapxrep}

\begin{proof}
  \label[proof]{prf:product-semantics}
  Suppose \( \processor{}: M \leadsto P \) and
  \( \processortwo{}: N \leadsto Q \) are stream processors with
  \( \processor{} = (S, f, s_{\id}, p_{\id}) \) and
  \( \processortwo{} = (T, g, t_{\id}, q_{\id}) \).
  Let \( \processor{} \dprod \processortwo{} = (S \times T, h, (s_{\id},
  t_{\id}), (p_{\id}, q_{\id})) \), and
  let \( h \) be defined as in \cref{def:parallel}. We verify
  \begin{align*}
    \denote{\processor{} \dprod \processortwo{}}((m, n))
     & = \run{}(\processor{} \dprod \processortwo{})((m, n))                   \\
     & = (p_{\id}, q_{\id}) \mul \out(h((m, n))((s_{\id}, t_{\id})))           \\
     & = (p_{\id}, q_{\id}) \mul (\out(\alpha(s_{\id})), \out(\beta(t_{\id})))
    \text{ where \( (\alpha, \beta) = (f \dprod g)((m, n)) \) }                \\
     & = (p_{\id}, q_{\id}) \mul (\out(f(m)(s_{\id})), \out(g(n)(t_{\id})))
    \tag{by definition of \( f \dprod g \)}                                    \\
     & = (p_{\id}, q_{\id}) \mul (p_{\id \leadsto m}, q_{\id \leadsto n})      \\
     & = (p_{m}, q_{n})                                                        \\
     & = (\denote{\processor{}}(m), \denote{\processortwo{}}(n))               \\
     & = (\denote{\processor{}} \dprod \denote{\processortwo{}})((m, n))
  \end{align*}
\end{proof}

\begin{proofsketch}
  Expand out the definitions (see the
  \hyperref[prf:product-semantics]{Appendix}).
\end{proofsketch}

\noindent
We now define how to split and merge streams. In particular, we define
splitters by their semantics.

\begin{definition}
  Let \( M \) be a commutative monoid. We define \( \pmerge : M \dprod M \to M \)
  as \
  \[
    \pmerge((a, b)) \defeq a + b.
  \]
  Recall that \( a + b \) is the conventional notation for the binary operator in
  a commutative monoid.
\end{definition}

\begin{proposition}
  \( \pmerge \) is a homomorphism.
\end{proposition}

\begin{definition}
  \label[definition]{def:splitter}
  A stream processor \( \processor{} : M \leadsto M \dprod M \) is a
  \textit{splitter} if we have
  \
  \[
    \processor{} \pipe \pure \pmerge \sim \pure \idfn
  \]
\end{definition}

\begin{remark}
  Splitters are not necessarily homomorphisms. For example, consider a stream
  processor \( \processor{} : \List[\Nat] \to \List[\Nat] \dprod \List[\Nat] \).
  If \( \processor{} \) should send elements to its outputs in round-robin
  fashion, it would need internal state to keep track of which output it should
  next send an input, so it would not be a homomorphism. On the other hand, if \(
  \processor{} \) should send even numbers to the first output and odd numbers to
  the second output, it could operate statelessly and would be a homomorphism.
\end{remark}

\begin{proposition}
  \label[proposition]{prop:splitmerge}
  Let \( \psplit \) be a splitter and \( f : M \to M \) be a homomorphism.
  We have
  \
  \[
    \pure f \sim \psplit \pipe (\pure f \dprod \pure f) \pipe \pure \pmerge
  \]
\end{proposition}

\begin{proof}
  Observe \
  \begin{align*}
    \denote{\pure \pmerge \pipe \pure f}((a, b))
     & = f(a + b) = f(a) + f(b)                                        \\
     & = \denote{\pure \pmerge}((f(a), f(b)))                          \\
     & = \denote{(\pure f \dprod \pure f) \pipe \pure \pmerge}((a, b))
  \end{align*}
  Thus we have
  \
  \begin{align*}
    \pure f
     & \sim \pure \idfn \pipe \pure f                                                       \\
     & \sim \psplit \pipe \pure \pmerge \pipe \pure f             \tag{\Cref{def:splitter}} \\
     & \sim \psplit \pipe (\pure f \dprod \pure f) \pipe \pure \pmerge
  \end{align*}
\end{proof}

\noindent
The right processor in \propref{prop:splitmerge} can be interpreted as
performing the computation \( f \) in parallel across two processors.
\Propref{prop:splitmerge} and \cref{cor:decompose} together model the
\textit{MapReduce} programming framework; given a stateful operator \(
\processor{} = (S, f, s_\id, o_\id) \), we can rewrite it to
\[
  \processor{} \sim \pure f \pipe \eval \processor{}                                        \\
  \sim \psplit \pipe (\pure f \dprod \pure f) \pipe \pure \pmerge \pipe \eval \processor{}.
\]
The rightmost program computes the pure computation for \( \processor{} \)
across two independent processors. In order to do useful work in parallel, the
monoidal product on \( \Img{f} \) must be ``efficiently'' parallelizable; see
\cref{sec:partial-eval}. For general processors \( \processor{} \), all work is
delayed until the evaluation.

Many stream programs (for example, distributed databases) need to perform
operations on \textit{all pairs} of two input streams. In order to model these
programs, we define the tensor product of two commutative monoids \( M \) and
\( N \).

\begin{definition}
  Let \( M \) and \( N \) be commutative monoids. The \textit{tensor product} \(
  M \tensor N \) is the formal product of elements \( m \tensor n \) for all \( m
  \in M \) and \( n \in N \) with identity \( \id \) subject to the relations:
  \begin{align*}
    \id_M \tensor n       & = \id = m \tensor \id_N \tag{identity}                 \\
    t_1 + t_2             & = t_2 + t_1 \tag{commutativity}                        \\
    (m_1 + m_2) \tensor n & = (m_1 \tensor n) + (m_2 \tensor n) \tag{left-linear}  \\
    m \tensor (n_1 + n_2) & = (m \tensor n_1) + (m \tensor n_2) \tag{right-linear}
  \end{align*}
\end{definition}

\noindent
The tensor product \( M \tensor N \) is a monoid that models all pairs of the
elements of \( M \) and \( N \). For example, the tensor product \( (m_1 + m_2
+ m_3) \tensor (n_1 + n_2 + n_3) \) expands to the sum \( \sum_{i=1}^{3}
\sum_{j=1}^{3} m_i \tensor n_j \) with nine terms (by the left-linear and
right-linear relations).

\begin{definition}
  Let \( \mathrm{Pairs} : (M \dprod N) \to (M \tensor N) \) be the function \(
  \mathrm{Pairs}((m, n)) \defeq m \tensor n \).
\end{definition}

\begin{proposition}
  \( \mathrm{Pairs} \) is not a homomorphism, but is a stream function.
\end{proposition}

\begin{proof}
  \( \mathrm{Pairs} \) is not a homomorphism as
  \[
    \mathrm{Pairs}((m_1, n_1)) = m_1 \tensor n_1
    \neq
    \id = m_1 \tensor \id + \id \tensor n_1 = \mathrm{Pairs}((m_1, \id)) + \mathrm{Pairs}((\id, n_1)).
  \]

  \noindent
  However, \( \mathrm{Pairs} \) is a stream function; define \(
  \Update{\mathrm{Pairs}}((m, n), (m', n')) = m \tensor n' + m' \tensor n + m'
  \tensor n' \). It is easy to check that \( \Update{\mathrm{Pairs}} \) satisifes
  the required properties in \cref{def:stream-function}.
\end{proof}

\begin{definition}
  Let \( \pairs : (M \dprod N) \leadsto (M \tensor N) \) be a stream processor
  such that \( \denote{\pairs} = \mathrm{Pairs} \). (Such a processor must exist
  by \fullref{thm:completeness}.)
\end{definition}

\subsection{Application: correctness of database join optimization}
\label{sec:join}

We will now prove the correctness of the database join optimizations we
outlined in \cref{sec:example}. Recall that one optimization we performed was
partitioning a single processor that computed the entire join into two
processors which each computed a subset of the join in parallel. The main step
in the proof of correctness for partitioning will require that processors can
fully evaluate their output independently of each other. Hence, we define
\textit{independent inputs} as follows:

\begin{definition}
  \label[definition]{def:independent}
  Suppose \( \processor{} : M \leadsto N \) is a stream processor with empty
  initial output. We say that \( a, b \in M \) are independent inputs to
  \( \processor{} \) if we have
  \
  \[
    \denote{\processor{}}(ab) = \denote{\processor{}}(a) \mul \denote{\processor{}}(b)
  \]
\end{definition}

\noindent
Observe that if \( F : M \to N \) is a homomorphism, then all of its inputs are
independent. But even if \( F : M \to N \) is not a homomorphism, it may behave
like a homomorphism for specific inputs \( a, b \in M \).
\Cref{def:independent} states this precisely.

Recall that \propref{prop:splitmerge} allowed us to parallelize pure
computations. The next lemma additionally allows us to parallelize evaluations
when inputs are independent.

\begin{lemmaapxrep}
  \label{lem:independentmerge}
  Let \( \processortwo{} : M \leadsto N \dprod N \)
  and
  \( \processor{} : N \leadsto P \)
  be stream processors such that for all outputs
  \( (n_1, n_2) \in \Img{\denote{\processortwo{}}} \),
  \( n_1 \) and \( n_2 \) are
  independent inputs to \( \processor{} \). Then we have
  \
  \[
    \processortwo{} \pipe \pure \pmerge_N \pipe \processor{} \sim
    \processortwo{} \pipe (\processor{} \dprod \processor{}) \pipe \pure \pmerge_P.
  \]
\end{lemmaapxrep}

\begin{proof}
  \label{prf:independentmerge}
  We need to show that for all \( m \in M \),
  \
  \[
    \denote{\processortwo{} \pipe \pure \pmerge_N \pipe \processor{}}(m) =
    \denote{\processortwo{} \pipe (\processor{} \dprod \processor{}) \pipe \pure \pmerge_P}(m).
  \]
  We have
  \begin{align*}
    \denote{\processortwo{} \pipe \pure \pmerge_N \pipe \processor{}}(m)
     & = \denote{\processor{}}(\pmerge_N(\denote{\processortwo{}}(m)))                                \\
     & = \denote{\processor{}}(\pmerge_N((n_1, n_2)))                                                 \\
     & = \denote{\processor{}}(n_1 + n_2)                                                             \\
     & = \denote{\processor{}}(n_1) + \denote{\processor{}}(n_2)   \tag{\cref{def:independent}}       \\
     & = \pmerge_P((\denote{\processor{}}(n_1), \denote{\processor{}}(n_2))                           \\
     & = \pmerge_P((\denote{\processor{}} \dprod \denote{\processor{}})((n_1, n_2)))                  \\
     & = \pmerge_P((\denote{\processor{}} \dprod \denote{\processor{}})(\denote{\processortwo{}}(m))) \\
     & = \denote{\processortwo{} \pipe (\processor{} \dprod \processor{}) \pipe \pure \pmerge_P}(m)   \\
    \begin{aligned}
      \text{where } & n_1 = p_1(\denote{\processortwo{}}(m)) \in p_1(\Img{\denote{\processortwo{}}})  \\
      \text{and }   & n_2 = p_2(\denote{\processortwo{}}(m)) \in p_2(\Img{\denote{\processortwo{}}}). \\
    \end{aligned}
  \end{align*}
\end{proof}

\noindent
The proof of \lemref{lem:independentmerge} is simple and can be found in the
\hyperref[prf:independentmerge]{Appendix.} \\

\emph{The proof.}
First, we optimize the composition \( \pairs \pipe \pure \filter \).
We interpret \( \pairs \pipe \pure \filter \) as
the composition of two processors where \( \pairs \) computes all \( mn \)
pairs of its input and sends each to \( \pure \filter \) to be filtered
out. We want to optimize the composition into a single processor,
\( \mathrm{join} \), which immediately discards tuples that do not match the
filter. Thus we have
\
\begin{align*}
  \pairs \pipe \pure \filter
   & \sim \pure f_{\pairs} \pipe \eval \pairs \pipe \pure \filter
  \tag{\cref{cor:decompose}}                                                            \\
   & \sim \pure f_{\pairs} \pipe \pure \mathrm{filter_*} \pipe \eval_{\filter_*} \pairs
  \tag{\cref{cor:exchangeeval}}                                                         \\
   & \sim \pure{} (f_{\pairs} \pipe \mathrm{filter_*}) \pipe \eval_{\filter_*} \pairs
  \tag{\cref{cor:decoupling}}                                                           \\
   & \sim \mathrm{join}
  \tag{\propref{prop:fusion}}
\end{align*}

where \( \mathrm{join} \) is the fused processor that is equivalent to \(
\pure{} (f_{\pairs} \pipe \mathrm{filter_*}) \pipe \eval_{\filter_*} \pairs \).

Next, we prove that partitioning \( \mathrm{join} \) into two parallel
processors is correct. We prove this for general commutative monoids \( M \)
and \( N \) generated by sets \( A \) and \( B \) respectively; for the
specific case of e.g. databases with set semantics, one can instantiate \( M =
\Set{U \times V} \) and \( N = \Set{V \times W} \).

\begin{theorem}
  \label[theorem]{thm:join-partition}
  Let \( M \) and \( N \) be commutative monoids generated by sets \( A \) and \( B \),
  respectively. Let \( \mathrm{hash}_A : A \to \{ 0, 1 \} \),
  \( \mathrm{hash}_B : B \to \{ 0, 1 \} \), and
  \( \mathrm{pred} : A \times B \to \{ \top, \bot \} \) be functions
  such that
  \( \forall a \in A, b \in B.\;
  \mathrm{hash}_A(a) \neq \mathrm{hash}_B(b) \implies
  \mathrm{pred}(a, b) = \bot \).
  Define \( \psplit : M \dprod N \to (M \dprod N) \dprod (M \dprod N) \) as a
  homomorphism on the generators by \
  \[
    \psplit((a, \id)) = \begin{cases}
      ((a, \id), \id) & \text{ if \( \mathrm{hash_A}(a) = 0 \)} \\
      (\id, (a, \id)) & \text{ if \( \mathrm{hash_A}(a) = 1 \)} \\
    \end{cases} \hspace{2em}
    \psplit((\id, b)) = \begin{cases}
      ((\id, b), \id) & \text{ if \( \mathrm{hash_B}(b) = 0 \)}  \\
      (\id, (\id, b)) & \text{ if \( \mathrm{hash_B}(b) = 1 \)}. \\
    \end{cases}
  \]
  \
  and \( \filter : M \tensor N \to M \tensor N \) as a homomorphism
  on the generators by
  \
  \[
    \filter(a \tensor b) = \begin{cases}
      a \tensor b & \text{if \( \mathrm{pred}(a, b) = \top \)} \\
      \id         & \text{if \( \mathrm{pred}(a, b) = \bot \)}.
    \end{cases}
  \]
  Then we have \
  \[
    \mathrm{join}
    \sim \pure \psplit \pipe (\mathrm{join} \dprod \mathrm{join}) \pipe \pure \pmerge.
  \]
\end{theorem}

\begin{proof}
  It is easy to check that \( \pure \psplit \) is a splitter. Thus we have the
  equivalences
  \begin{align*}
    \mathrm{join}
     & \sim \pure \psplit \pipe \pure \pmerge \pipe \mathrm{join} \\
     & \sim \pure \psplit
    \pipe (\mathrm{join} \dprod \mathrm{join})
    \pipe \pure \pmerge
    \tag{\Lemref{lem:independentmerge}}
  \end{align*}
  where the last equivalence follows if \( \pure \mathrm{split} \) and \(
  \mathrm{join} \) satisfy the conditions of independent inputs from
  \lemref{lem:independentmerge}. To check this, let \( ((m_1, n_1), (m_2, n_2))
  \in \Img{\mathrm{split}} \). Then we have \
  \begin{align*}
      & \; \denote{\mathrm{join}}((m_1, n_1) + (m_2, n_2))                         \\
    = & \; \denote{\pairs \pipe \pure \filter}((m_1, n_1) + (m_2, n_2))            \\
    = & \; \filter(\denote{\pairs}((m_1, n_1) + (m_2, n_2)))                       \\
    = & \; \filter(m_1 \tensor n_1) + \filter(m_1 \tensor n_2)
    + \filter(m_2 \tensor n_1) + \filter(m_2 \tensor n_2)                          \\
    = & \; \filter(m_1 \tensor n_1) + \id + \id + \filter(m_2 \tensor n_2) \tag{*} \\
    = & \; \denote{\pairs \pipe \pure \filter}((m_1, n_1))
    +   \denote{\pairs \pipe \pure \filter}((m_2, n_2))                            \\
    = & \; \denote{\mathrm{join}}((m_1, n_1))
    +   \denote{\mathrm{join}}((m_2, n_2))                                         \\
  \end{align*}
  where step (*) follows because if we expand, for example,
  \( m_1 \tensor n_2 \) into the sum of generators \( \sum_{i,j} a_i \tensor b_j \),
  we have \( \mathrm{hash}_A(a_i) = 0 \) and \( \mathrm{hash}_B(b_j) = 1 \) for all
  \( i, j \), so \( \mathrm{pred}(a_i, b_j) = \bot \) by hypothesis. This implies
  \( \filter((a_i, b_j)) = \id \) and
  \( \filter((m_1, n_2)) = \id \). By a similar argument,
  \( \filter((m_2, n_1)) = \id \). We conclude that the conditions
  for \lemref{lem:independentmerge} are satisfied.
\end{proof}

\begin{remark}
  Observe that to check the conditions for \lemref{lem:independentmerge} we used
  the equivalence proved earlier, \( \denote{\mathrm{join}} = \denote{\pairs
    \pipe \pure \filter} \), as we can more easily reason about the latter's
  semantics.
\end{remark}

\subsection{Ticks}
\label{sec:ticks}

Some streaming frameworks that operate over lattices feature support for
\textit{stratified negation} \citep{dedalus}; that is, they asynchronously
compute a monotonic function on inputs within a stratum, but support negation
across strata. An emblematic program for these frameworks is a program which
computes stratified set difference, which we model using the free product.

\newcommand{\fprod}{\star}

\begin{definition}
  \label[definition]{def:free-product}
  Let \( M \) and \( N \) be monoids. We define the free product of \( M \) and \( N \),
  written \( M \fprod N \), as formal products of the form \( m \fprod \_ \) for
  \( m \in M \) and \( \_ \fprod n \) for \( n \in N \) subject to the relations
  \begin{align*}
    (\id_M \fprod \_ ) = \id             & = (\_ \fprod \id_N)  \tag{identity} \\
    (m_1 \fprod \_ )\mul (m_2 \fprod \_) & = (m_1m_2 \fprod \_) \tag{left}     \\
    (\_ \fprod n_1) \mul (\_ \fprod n_2) & = (\_ \fprod n_1n_2) \tag{right}
  \end{align*}
\end{definition}

\begin{remark}
  When \( M \) and \( N \) are disjoint, we may simply write elements in \( M
  \fprod N \) as formal products of the form \( m_1 n_1 m_2 n_2 \ldots \) where
  \( m_i \in M \) and \( n_j \in N \).
\end{remark}

\begin{remark}
  The direct product (\cref{def:direct-product}) is the categorical product,
  while the free product (\cref{def:free-product}) is the categorical coproduct.
\end{remark}

\noindent
We now consider set difference. First, we observe that non-stratified set
difference is not expressible in our semantics. Then, we model stratified set
difference in two different ways.

\begin{proposition}
  Let \( A \) be a nonempty set. The function \( F : \Set{A} \dprod \Set{A} \to
  \Set{A} \) which computes the set difference \( F((A, B)) = A \setminus B \) is
  not a stream function.
\end{proposition}

\begin{proof}
  Let \( a \in A \) and define \( p = (\{ a \}, \emptyset) \) and \( n =
  (\emptyset, \{ a \}) \). Then, \( p + n = (\{ a \}, \{ a \}) \) and for any \(
  \Update{F} \), there is no \( x \in \Set{A} \) which allows \( \Update{F} \) to
  satisfy \cref{def:stream-function}: \
  \[
    F(p + n) = \emptyset < \{ a \} + x = F(p) + \Update{F}(p, n)
  \]
\end{proof}

\begin{remark}
  In the above proof, \( p \) stands for \textit{positive} and \( n \) stands for
  \textit{negative}. For the set difference \( A \setminus B \), we refer to \( A
  \) as the positive set and \( B \) as the negative set.
\end{remark}

\newcommand{\Bool}{\mathbb{B}}
\newcommand{\tick}{\mathbin{\top}}
\newcommand{\Tick}{\mathbb{T}}

\begin{definition}
  Let \( \Bool \) be the commutative monoid \( (\{ \bot, \top \}, \vee, \bot) \)
  where \( \top \vee x = \top \) for all \( x \in \Bool \).
\end{definition}

\noindent
For any monoid \( M \), define \( \Tick[M] = M \fprod \Bool \). We interpret
the monoid \( \Tick[M] \) as \( M \) with strata. We will use \( \_ \fprod
\tick \) (pronounced ``tick'' in this context) to denote the boundary between
two strata. As mentioned above, we will abbreviate \( \_ \fprod \tick \) to
simply \( \tick \) when unambiguous.

\begin{proposition}
  Let \( A \) be a set. Define \( I = \Tick[\Set{A} \dprod \Set{A}] \) and \( O =
  \Tick[\Set{A}] \). Then the function \( F : I \to O \) given by \
  \[
    F((A_1, B_1)
    \tick (A_2, B_2) \tick (A_3, B_3) \tick \ldots) = (A_1 \setminus \emptyset)
    \tick (A_2 \setminus B_1) \tick (A_3 \setminus B_2) \tick \ldots
  \] is a stream function.
\end{proposition}

\begin{proofsketch}
  It is easy to check that \( \Update{F} : I \times I \to O \) defined on the
  generator \( x \in I \) \
  \[
    \Update{F}(\ldots \tick (A_1, B_1) \tick (A_2, B_2),
    x) = \begin{cases}
      \id                & \text{if \( x = \id \)}          \\
      \tick              & \text{if \( x = \tick \)}        \\
      A_2' \setminus B_1 & \text{if \( x = (A_2', B_2') \)} \\
    \end{cases}
  \]
  satisfies the conditions in \cref{def:stream-function}.
\end{proofsketch}

\noindent
Another way of modeling stratified set difference is as a stream function which
accepts a \textit{list} of set pairs and outputs a list of stratified
differences.

\begin{proposition}
  Let \( A \) be a set. Define \( I' = \List[\Set{A} \dprod \Set{A}] \) and \( O'
  = \List[\Set{A}] \). Then the function \( G : I' \to O' \) given by \
  \[
    G([(A_1, B_1), (A_2, B_2), (A_3, B_3), \ldots])
    = [(A_1 \setminus \emptyset), (A_2 \setminus B_1), (A_3 \setminus B_2), \ldots]
  \]
  is a stream function.
\end{proposition}

\begin{proofsketch}
  Define \( \Update{G} : I' \times I' \to O' \) on the generator \( [(A_2, B_2)]
  \in I' \) by \
  \[
    \Update{O}([\ldots, (A_1, B_1)], [(A_2, B_2)]) = [A_2 \setminus B_1]
  \]
  Again, it is easy to check that \( \Update{G} \) satisfies condition~(1) in
  \cref{def:stream-function}. Because \( O' \) is a left-cancellative monoid,
  condition~(2) is automatically satisifed by
  \propref{prop:left-cancel-stream-function}.
\end{proofsketch}

\noindent
The main difference between \( F \) and \( G \) is in their incremental
behavior. A stream processor that implements \( F \), upon input \( (a, \id)
\in I \), can emit \( (a, \id) \) immediately (if \( a \) is not in its current
negative set). On the other hand, a stream processor that implements \( G \)
must wait for its next input \( [(A, B)] \in I' \) to become ``sealed'' before
it can emit any output.

Both models are useful in different situations. If a stream processor should be
able to emit output asynchronously within a stratum, then one should model its
input as \( \Tick[M] \). On the other hand, if a stream processor cannot emit
any useful output until a tick arrives, it may be more convenient to model its
input as a list.

\subsection{Feedback composition}
\label{sec:loops}

Many stream programs form dataflow cycles; that is, they feed their output back
into an input. For example, to implement recursive inference, Datalog engines
send the computed set of output facts back into an input until the output
reaches a fixed point. To model these programs, we define the feedback
composition of stream functions and stream processors.

Suppose \( F : M \times U \to N \times U \) is a stream function. Fix some
input \( m \in M \), and consider the sequence of tuples recursively defined by
\
\begin{align}
  \label{eqn:fix-sequence}
  x_0 = F(\id),                                         \hspace{0.5em}
  x_1 = \Update{F}(\id, y_1),                           \hspace{0.5em}
  x_2 = \Update{F}(y_1, y_2),                           \hspace{0.5em}
  \ldots                                                \hspace{0.5em}
  x_{i} = \Update{F}(y_1 y_2 \ldots y_{i-1}, y_{i})     \hspace{0.5em}
  \tag{fix sequence} \\
  \text{where } (n_i, u_i) = x_i \text{ and } y_i = (m, u_{i-1}) \nonumber
\end{align}

\noindent
That is, we keep feeding \( F \)'s feedback output \( u_i \) back into its
input (replacing \( n_i \) with \( m \) each time; see the definition of \( y_i
\)). One would hope that this process reaches a fixed point for \( u \): i.e.
there is some index \( k \in \Nat \) such that \( u_{k+1} = u_{k} \). Then we
could simply emit \( n_{k+1} \). Unfortunately, such a fixed point does not
always exist. A simple example is the stream function \( F : \Nat \dprod \Nat
\to \Nat \dprod \Nat \) given by \( F(m, u) = (u, u + m) \), which reaches a
fixed point if and only if \( m = 0 \).

Even if a fixed point does exist for every \( m \in M \), there is an
additional complication we must account for if we want to model real stream
programs. Suppose we have a stream processor \( \processor{} \) whose semantics
is the fixed point of \( F \). What happens if additional input \( m' \)
arrives while \( \processor{} \) is in the middle of computing the fixed point
for \( m \)? Should \( \processor{} \) block until it is finished computing the
fixed point for \( m \)? Or, can it accept \( m' \) immediately and continue
its fixed point computation? Some ``fully asynchronous'' stream programs (for
example, monotonic Datalog) can accept \( m' \) immediately and will still
compute the same result. Under what conditions is it correct to accept \( m' \)
immediately? This is a nontrivial question and we leave it to future work.

Instead, we define a stratified approximation of the above fixed point. Above,
we fixed a single \( m \) and ran the fixed point computation on \( m \). In
the stratified approximation, the operator ingests \( m_1 \); then it computes
one round of the fixed point; then it ingests \( m_2 \); computes another
round, etc.

\begin{definition}
  \label[definition]{def:loop-semantics}
  Let \( F : M \dprod U \to N \dprod U \) be a stream function.
  Fix a sequence \( m_1, m_2, \ldots, m_k \) of inputs and consider the
  sequence of tuples
  \
  \begin{align*}
    x_0 = F(\id),                                         \hspace{0.5em}
    x_1 = \Update{F}(\id, y_1),                           \hspace{0.5em}
    x_2 = \Update{F}(y_1, y_2),                           \hspace{0.5em}
    \ldots                                                \hspace{0.5em}
    x_{i} = \Update{F}(y_1 y_2 \ldots y_{i-1}, y_{i})     \hspace{0.5em}
    \\
    \text{where } (n_i, u_i) = x_i \text{ and } y_i = (m_i, u_{i-1})
  \end{align*}
  We define \( \loopfn F : \List[M] \to \List[N] \) to be the stream function
  where
  \
  \[
    (\loopfn F)([m_1, m_2, \ldots, m_k]) \defeq [n_0, n_1, n_2, \ldots, n_k].
  \]
\end{definition}

\noindent
Note the only difference between the sequence in \cref{def:loop-semantics} and
\cref{eqn:fix-sequence} is in the definition of \( y_i \), where instead of
holding \( m \) constant, we give the loop a ``synchronization point'' to
ingest \( m_i \). If we set \( m_1 = m_2 = \ldots = m_k = m \) then we recover
the fixed point definition, where each arrival of \( m \) signals the loop to
execute another round in the fixed point computation.

\begin{example}
  We can model monotone Datalog using \( \loopfn \). Let \( E \) and \( I \) be
  monoids representing the extensional database and the intensional database,
  respectively. The immediate consequence operator can be modeled as a (monotone)
  stream function \( T_P : E \dprod I \to I \dprod I \), and the (recursive)
  Datalog program is simply \( \loopfn T_P : \List[E] \to \List[I] \). Note that
  \( \List[E] \) are \textit{batches} of the extensional database; depending on
  the specific implementation of \( T_P \), these batches may need to be
  monotonically increasing to ensure the output is correct.
\end{example}

\begin{definition}
  Suppose \( \processor{} : M \dprod U \leadsto N \dprod U \) is a stream
  processor with \( \processor{} = (S, f, s_\id, o_\id) \). Let \( (n_0, u_0) =
  o_\id \). Define the stream processor \( \loopfn \processor{}: \List[M]
  \leadsto \List[N] \defeq (S \times U, g, (s_\id, u_0), [n_0]) \) where \( g :
  \List[M] \to \State{S \times U}{\List[N]} \) is defined on the generators by
  \[
    g([m]) = (s, u) \mapsto \text{let } (s', (n, u')) = f(m, u)(s) \text{ in } ((s', u'), [n])
  \]
  Observe \( g \) is a homomorphism as we defined it on the generators of the
  free monoid.
\end{definition}

\noindent
The \( \loopfn \) combinator on a processor \( \processor{} : M \dprod U
\leadsto N \dprod U \) constructs a new processor that holds internal state \(
S \times U \). When it receives a batch \( m_j \), it feeds input \( m_j \) and
feedback \( u_{j-1} \) to \( \processor{} \), yielding the result \( n_i \). It
then updates its internal state to \( \processor{} \)'s new state and the new
feedback \( u_j \).

As usual, the following proposition connects \( \loopfn \)'s syntax to the
semantics.

\begin{proposition}
  Let \( \processor{} : M \dprod U \leadsto N \dprod U \) be a stream processor.
  Then
  \[
    \denote{\loopfn \processor{}} = \loopfn\, \denote{\processor{}}
  \]
\end{proposition}

\begin{proof}
  Let \( K_j = [m_1, m_2, \ldots, m_j] \) with \( K_0 = \id \) and \( n_j \) and
  \( u_j \) be as defined in \cref{def:loop-semantics} with \( F =
  \denote{\processor{}} \) and \( \Update{F}(p, a) = f(a)(s_p) \). Define \( G_j
  = f(m_j, u_{j-1})(s_{K_{j-1}}) \). That is, \( G_j \) is a tuple containing the
  updated state and incremental output of \( \processor{} \) upon input \( m_j \)
  and feedback \( u_{j-1} \) after having ingested the prefix \( K_{j-1} \). We
  verify \
  \begin{align*}
    \denote{\loopfn \processor{}}([m_1, m_2, \ldots, m_k])
     & = [n_0] \concat [p_1(\out(G_1))] \concat [p_1(\out(G_2))] \concat \ldots \concat [p_1(\out(G_k))]                                    \\
     & = [n_0] \concat [p_1(o_{0 \leadsto K_1})] \concat [p_1(o_{K_1 \leadsto K_2})] \concat \ldots \concat [p_1(o_{K_{k-1} \leadsto K_k})] \\
     & = [n_0] \concat [n_1] \concat [n_2] \concat \ldots \concat [n_k]                                                                     \\
     & = (\loopfn \denote{\processor{}})([m_1, m_2, \ldots, m_k])
  \end{align*}
\end{proof}

\noindent
The next two propositions say that we can ``extract'' pure processors from a
loop if they don't interact with the loop's feedback channel (or, in the other
direction, we can ``pull'' them in).

\begin{proposition}[Left tightening]
  \label[proposition]{prop:left-tightening}
  Let \( \processor{} : M \dprod U \leadsto N \dprod U \) be a stream
  processor and \( f : L \to M \) be a homomorphism. Then we have
  \[
    \pure f \pipe \loopfn \processor{}
    \sim \loopfn\, (\pure\, (f \dprod \idfn) \pipe \processor{}).
  \]
\end{proposition}

\begin{propositionapxrep}[Right tightening]
  \label{prop:right-tightening}
  Let \( \processor{} : M \dprod U \leadsto N \dprod U \) be a stream
  processor and \( g : N \to P \) be a homomorphism. Then we have
  \[
    \loopfn \processor{} \pipe \pure g
    \sim \loopfn\, (\processor{} \pipe \pure\, (g \dprod \idfn)).
  \]
\end{propositionapxrep}

We prove \propref{prop:right-tightening} in the
\hyperref[prf:right-tightening]{Appendix.} The proof of
\propref{prop:left-tightening} is extremely similar.

\begin{proof}
  \label{prf:right-tightening}
  Fix some input \( K = [m_1, m_2, \ldots, m_k] \) and let
  \( F = \denote{\processor{}} \) and \( \Update{F}(p, a) = f_{\processor{}}(a)(s_p) \).
  By \cref{def:loop-semantics}
  the semantics
  \( \denote{\loopfn \processor{} \pipe \pure g}(K) = (\denote{\loopfn \processor{}} \pipe g)(K) \)
  is the sequence
  \( [g(n_0), g(n_1), g(n_2), \ldots, g(n_k)] \) satisfying
  \
  \begin{align*}
    x_0 = F(\id),                                         \hspace{0.5em}
    x_1 = \Update{F}(\id, y_1),                           \hspace{0.5em}
    x_2 = \Update{F}(y_1, y_2),                           \hspace{0.5em}
    \ldots                                                \hspace{0.5em}
    x_{i} = \Update{F}(y_1 y_2 \ldots y_{i-1}, y_{i})     \hspace{0.5em}
    \\
    \text{where } (n_i, u_i) = x_i \text{ and } y_i = (m_i, u_{i-1})  \nonumber
  \end{align*}
  On the other hand, the semantics
  \( \denote{\loopfn\, (\processor{} \pipe \pure\, (g \dprod \idfn))}(K)
  = (\loopfn\, \denote{\processor{} \pipe \pure\, (g \dprod \idfn)})(K)  \)
  is the sequence \( [n_0', n_1', n_2', \ldots, n_k'] \) satisfying
  \begin{align*}
    x_0' = G(\id),                                          \hspace{0.5em}
    x_1' = \Update{G}(\id, y_1'),                           \hspace{0.5em}
    x_2' = \Update{G}(y_1', y_2'),                          \hspace{0.5em}
    \ldots                                                  \hspace{0.5em}
    x_{i}' = \Update{G}(y_1' y_2' \ldots y_{i-1}', y_{i}')  \hspace{0.5em}
    \\
    \text{where } (n_i', u_i') = x_i' \text{ and } y_i' = (m_i', u_{i-1}')  \nonumber
  \end{align*}
  where \( G = F \pipe (f \dprod \idfn) \) and
  \( \Update{G} = \Update{F} \pipe (f \dprod \idfn) \).
  It is clear that for all \( j \), we have \( n_j' = g(n_j) \).

\end{proof}

\subsection{Application: Transmission Control Protocol}
\label{sec:tcp}

Using feedback composition, we sketch a simplified model of TCP with
retransmission. We also sketch a proof that the receiver can reconstruct the
sender's stream, assuming eventual delivery. (We only model reliable delivery;
we do not model congestion control.)

\newcommand{\Pack}{\mathrm{Pack}}
\newcommand{\Ret}{\mathrm{Re}}

Let \( T \) be a set. Define the monoids \( \Pack = \Tick[\List[\Nat \times T]]
\) and \( \Ret = \Tick[\List[\Nat]] \). (Note that \( \times \) in \( \Pack \)
is the Cartesian product, not the direct product.) We interpret elements of \(
T \) as data to be sent over the network, \( \Pack \) as an ordered sequence of
data packets with sequence numbers, and \( \Ret \) as retransmission requests.
Ticks represent the passage of logical time.

Next, we define stream processors whose composition will represent a
(non-looped) TCP system: \ \begingroup\setlength{\jot}{0pt}
\begin{align*}
  \sender    : & \, \Tick[\List[T]] \dprod \Ret \leadsto \Pack \dprod \Ret                                                                                         \\
               & \, \text{stores in state all input data and their assigned sequence numbers}                                                                      \\
               & \, \text{on input \( ([m], \id) \), assigns the next sequence number \( i \) to \( m \) and emits \( ([(i, m)], \id) \) }                         \\
               & \, \text{on input \( (\id, [i]) \), retransmits packet \( i \), emitting \( ([(i, m)], \id) \) }                                                  \\
               & \, \text{on ticks \( (\tick, \id) \) or \( (\id, \tick) \), forwards the tick emitting \( (\tick, \id) \) }                                       \\[4pt]
  \receiver  : & \, \Pack \times \Ret \leadsto \List[T] \dprod \Ret                                                                                                \\
               & \, \text{stores in state all received data and the next unreceived sequence number \( n \)}                                                       \\
               & \, \text{on input \( ([(i, m_i)], \id) \), emits \( ([m_i, m_{i+1}, \ldots, m_{n' - 1}], [n']) \) if \( i = n \) and sets \( n \leftarrow n' \) } \\
               & \, \text{otherwise (including on ticks), sends retransmission request \( (\id, [n]) \) }                                                          \\[4pt]
  \network_1 : & \, \Pack \leadsto \Pack                                                                                                                           \\[4pt]
  \network_2 : & \, \Ret \leadsto \Ret
\end{align*}
\endgroup

\noindent
where \( \network_1 \) and \( \network_2 \) are processors which drop (replace
with \( \tick \)), reorder, and delay packets arbitrarily, except for each
packet \( p_i \) with sequence number \( i \) there is some \( t_i \in \Nat \)
such that if \( \network_1 \) or \( \network_2 \) processes at least \( t_i \)
packets it must immediately forward \( p_i \) upon receipt.

The composition of the above processors is the processor \( \processor{} :
\Tick[\List[T]] \dprod \Ret \leadsto \List[T] \dprod \Ret \): \
\[
  \processor{} = \sender \pipe\, (\network_1 \dprod \pure \idfn) \pipe \receiver \pipe\, (\pure \idfn \dprod \network_2)
\]
and the TCP system is given by the processor \( \loopfn \processor{} :
\Tick[\List[T]] \leadsto \List[T] \).

\begin{proposition}
  Fix some input sequence \( L = [t_1, t_2, \ldots, t_k] \). There exists some
  natural \( N \in \Nat \) such that \( \denote{\loopfn
    \processor{}}(L\underbrace{\top\top\ldots\top}_{N \text{ times}}) = L \).
\end{proposition}

\begin{proofsketch}
  By construction, each round of \( \loopfn \processor{} \) sends at least one
  message from the sender to the receiver and vice versa. Thus, on each round,
  the receiver requests retransmission of some packet \( p_i \) with sequence
  number \( i \). By assumption, both \( \network_1 \) and \( \network_2 \) must
  deliver \( p_i \) after \( t_i \) rounds; hence the sender will receive the
  request for retransmission by round \( t_i \), retransmit \( p_i \) by round \(
  t_i+1 \), and the receiver must receive \( p_i \). So a (loose) upper bound for
  \( N \) is \( N \leq \sum_{i = 1}^{k} t_i + 1 \).
\end{proofsketch}

\section{Efficient representations for State[S, N]}
\label{sec:partial-eval}

In \cref{def:stream-processor}, the homomorphism \( f : M \to \State{S}{N} \)
takes an input \( M \) to a monoid of functions \( \State{S}{N} \). But
generally, functions are harder to work with than data values---functions
generally cannot be serialized to be sent over e.g. a network, and there is
usually higher runtime cost to produce and reduce function abstractions. For
example, a typical implementation of the monoidal product \( \odot \) of \(
\State{S}{N} \) would create nested abstractions, evaluation of which require a
fully sequential reduction. Using partial evaluation, we can simplify
underneath the abstraction, but running a partial evaluator at runtime is also
expensive. Instead, we can find more efficient representations for \(
\State{S}{N} \) ahead of time by applying traditional compiler techniques.

\newcommand{\inj}{\phi}
\newcommand{\prj}{\psi}

Consider any homomorphism \( f : M \to \State{S}{N} \). We need only represent
the subset of functions \( \Img{f} \subseteq \State{S}{N} \). Then our goal is
to find a monoid \( P \) such that
\begin{enumerate}
  \item The elements of \( P \) are easily serializable
  \item The monoidal product on \( P \) is ``efficient''\footnote{ We deliberately
          leave this vague. Informally, an ``efficient'' monoidal product does meaningful
          work (unlike function composition in most programming languages, which simply
          delays doing the work of either function). }
  \item There exist homomorphisms \( \inj : \Img{f} \to P \) and \( \prj : P \to
        \Img{f} \) such that \( \inj \pipe \prj = \idfn_{\Img{f}} \).
        \begin{figure}[ht]
          \centering
          \begin{tikzcd}
            M \arrow[r, "f"]
            & \Img{f} \arrow[r, "\inj", bend left]
            & P \arrow[l, "\prj", bend left]
            & \hspace{-2em} \text{where } \inj \pipe \prj = \idfn_{\Img{f}}
          \end{tikzcd}
          \caption{Condition (3) for \( P \), a concrete representation of \( \Img{f} \)}
        \end{figure}
\end{enumerate}

\noindent
We say then that \( \Img{f} \) \textit{embeds} into \( P \). Then, concrete
implementations can manipulate and send elements of \( P \) over the network.
In most cases, given access to the source code of \( f \), we can construct \(
P \) using traditional compiler techniques. We demonstrate this by example.

\newcommand{\inv}[1]{\overline{#1}}
\newcommand{\bit}{B}

\begin{example}[Full adder]
  Let \( \bit = \{0, 1\} \) be bits. Let
  \( \processor{} : \List[\bit \times \bit] \leadsto \List[\bit] = (\bit, f, 0, \id) \)
  be the stream processor
  where \( f : \List[\bit \times \bit] \to \State{\bit}{\List[\bit]} \) is given by (on
  the generators):
  \
  \begin{align}
    \label[equation]{eq:adder}
    f([(a, b)]) =
    \begin{cases}
      \; c \mapsto (0, [c])       & \text{if \( (a, b) \) is \( (0, 0) \) }                \\
      \; c \mapsto (c, [\inv{c}]) & \text{if \( (a, b) \) is \( (0, 1) \) or \( (1, 0) \)} \\
      \; c \mapsto (1, [c])       & \text{if \( (a, b) \) is \( (1, 1) \) }
    \end{cases}
  \end{align}
  where \( \inv{0} = 1 \) and \( \inv{1} = 0 \). The application \( f([(a, b)])(c) \)
  models a full adder, where \( a \) and \( b \) are the input bits and \( c \) is
  the carry bit. Then \( \processor{} \) is a stream processor that computes
  the addition of two (arbitrary length) binary integers, whose digits
  are passed as pairs.
\end{example}

\subsection{Defunctionalization}

We can apply \textit{defunctionalization} \citep{reynolds72} to represent
functions as data. Let \( P_1 = \langle x, y, z \rangle \); i.e. \( P_1 \) is
the free monoid on three generators. Observe that \( \Img{f} \) is generated by
the functions in \cref{eq:adder}; hence we can define \( \inj_1 : \Img{f} \to P
\) by \
\begin{align*}
  \inj_1(\alpha) =
  \begin{cases}
    \; x & \text{if \( \alpha = c \mapsto (0, [c]) \)}       \\
    \; y & \text{if \( \alpha = c \mapsto (c, [\inv{c}]) \)} \\
    \; z & \text{if \( \alpha = c \mapsto (1, [c]) \)}
  \end{cases}
\end{align*}

\noindent
Define \( \prj_1 \) on the generators \( x, y, z \) as \( \inj_1^{-1} \). It is
clear that \( \inj_1 \pipe \prj_1 = \idfn_{\Img{f}} \). In a concrete
implementation, the product \( x \mul y \) of \( P_1 \) can be represented
efficiently, e.g., by constructing a new binary tree node with children \( x \)
and \( y \). Observe that elements of \( P_1 \) are simply data and can be sent
over the network; the receiving processor can apply \( \inj_1^{-1} \) to
recover elements of \( \State{\bit}{\List[\bit]} \).

Given source code for \( f \), defunctionalization allows us to construct \( P
\) in general.\footnote{ Note that for arbitrary homomorphisms \( M \to
  \State{S}{N} \), the generators of \( P \) may need to be indexed by data
  captured by the closure. In \cref{eq:adder}, we created abstractions within
  each branch to avoid capturing \( (a, b) \). } However, for \( \Img{f} \), this
encoding may not be the most efficient encoding possible: for products \( xyz
\in P \), the inverse \( \prj(xyz) \) is a product \( \alpha \odot \beta \odot
\gamma \in \State{\bit}{\List[\bit]} \), whose evaluation is a sequential
reduction (assuming we cannot reduce under abstractions).%

\subsection{Bounded static variation (The Trick)}
\label{sec:thetrick}

By applying a technique called bounded static variation, or the ``The
Trick''~\cite{jones1993partial} from partial evaluation, we can construct a
monoid \( P_2 \) whose product (unlike \( \odot \) or usual function
composition) allows both sides to potentially ``start working'' at the same
time.

The key observation is similar to that of speculative execution: we can start
computing a function's results before receiving its argument if we are willing
to guess the argument. Bounded static variation is useful when a function's
domain is small: then one can ``guess'' all possible arguments, precompute the
respective results, and then choose between them.
In our adder example, the state \( c \) (the carry bit) has only two possible
values. So, another representation of \( \alpha \in \Img{f} \) is an arity-two
tuple whose \( 0 \)-value and \( 1 \)-value hold the results \( \alpha(0) \)
and \( \alpha(1) \), respectively. Thus, let \( P_2 \) be the monoid whose
elements are tuples of the form \
\[
  \left( 0 \mapsto \bit \times \List[\bit], 1 \mapsto \bit \times \List[\bit] \right)
\]

\noindent
We need to also modify \( \odot \) for this new representation. For \( x, y \in
P_2 \), define \( x \odot y \) by \
\[
  \left(
  \hspace{0.3em}
  \begin{aligned}[c]
    0 \mapsto \; & \text{let } c = 0 \text{ in}         \\
                 & \text{let } (c, a) = x(c) \text{ in} \\
                 & \text{let } (c, b) = y(c) \text{ in} \\
                 & (c, a \doubleplus b)                 \\
  \end{aligned}
  \begin{aligned}[c]\; \\ \; \\ \; \\ \;\; , \;\;\; \\ \end{aligned}
  \begin{aligned}[c]
    1 \mapsto \; & \text{let } c = 1 \text{ in}         \\
                 & \text{let } (c, a) = x(c) \text{ in} \\
                 & \text{let } (c, b) = y(c) \text{ in} \\
                 & (c, a \doubleplus b)                 \\
  \end{aligned} \\
  \hspace{0.3em}
  \right)
\]
where the ``call'' syntax \( x(c) \) denotes the selection of the \( 0 \)-value
or the \( 1 \)-value of the tuple \( x \) if \( c = 0 \) or \( c = 1 \)
respectively (and same for \( y(c) \)).

Observe that \( \odot \) in \( P_2 \) is \( \odot \) in \(
\State{\bit}{\List[\bit]} \) specialized to concrete inputs for \( \bit \) (\(
c = 0 \), \( c = 1 \)). Finally, we define \( \inj_2 : \Img{f} \to P_2 \) and
\( \prj_2 : P_2 \to \Img{f} \) as: \
\begin{align*}
  \inj_2(\alpha)  = ( 0 \mapsto \alpha(0), 1 \mapsto \alpha(1) )
   &  & \prj_2(x) = c \mapsto x(c)
\end{align*}

\noindent
Then, for example, we have
\[
  (f \pipe \phi_2)([(0, 1), (0, 0), (1, 1)]) = (0 \mapsto (1, [1, 0, 0]), 1 \mapsto (1, [0, 1, 0]))
\]

\noindent
Beware the left pair \( (0, 1) \) in the list are the least significant bits.
The above equation states that the addition of the binary numbers \( 100 \) and
\( 101 \) is \( 001 \) with carry \( 1 \) if the initial carry was \( 0 \), and
\( 010 \) with carry \( 1 \) if the initial carry was \( 1 \).

We finally note that products in \( P_2 \) are about twice as expensive as
products in the original monoid \( \Img{f} \) followed by evaluation, as in \(
P_2 \), we always compute the output for both possible starting states. But \(
P_2 \) has an advantage in parallel systems since multiple processors can
independently compute products in \( P_2 \) and a merging processor can combine
those outputs quickly.

\begin{remark}
  Witness that if we apply the above construction to elements of \( \State{\{
    \unit \}}{N} \), then we obtain a monoid \( (( \unit \mapsto \unit \times N ),
  \odot, ( \unit \mapsto (\unit, \id) )) \) which is isomorphic to \( N \).
  Hence, for stateless processors, one possible representation is exactly the
  output monoid.
\end{remark}

\section{Related work}

\subsection{Stream transducers and incremental computation}

\defcitealias{mamouras}{Mamouras' \emph{Semantic foundations for deterministic dataflow and stream processing} \citeyearpar{mamouras}}
\defcitealias{mario-thesis}{Alvarez-Picallo's \emph{Change actions: from incremental computation to discrete derivatives} \citeyearpar{mario-thesis}}

Our definitions of {stream function}~(\ref{def:stream-function}) and {stream
    processor}~(\ref{def:stream-processor}) combine two crucial pieces of prior
work: \citetalias{mamouras} and \citetalias{mario-thesis}. We follow Mamouras
closely in using monoids to model streaming data and in defining two related
notions of stream transformers: one declarative and functional---Mamouras'
{stream transduction}~[p.~10, defn.~8], our {stream function}---the other a
state machine which can implement the first---Mamouras' {stream
transducer}~[p.~12, defn.~14], our {stream processor}.

However, our definitions are stricter than Mamouras', because we impose
Alvarez-Picallo's regularity conditions~[p.~27, defn.~3.1.3; our
\crefrange{eqn:regular-id}{eqn:regular-mul}]. This lets us replace Mamouras'
complex coherence conditions~[p.~14, defn.~20] with the well-known monoid
homomorphism laws. It also connects stream processing to the rich theory of
incremental computation explored in Alvarez-Picallo's thesis and related
work~\citep{DBLP:phd/dnb/Giarrusso20,DBLP:conf/esop/Alvarez-Picallo19,DBLP:conf/pepm/Liu24},
which can be exploited for efficiency optimizations. Finally, it reuses
functional constructions (the State/Writer monads), making it more suitable for
embedding into functional languages.

We must also note that our stream functions are an instance of
Alvarez-Picallo's differential maps~\citep[p.~31, defn.~3.1.5]{mario-thesis}.
Alvarez-Picallo's definition is more general because he considers sets equipped
with monoid actions; like Mamouras, we consider only monoids acting on
themselves. However, Alvarez-Picallo deals only with the declarative/functional
side and has no equivalent of stream transducers/processors.
We summarize these relations in this \( 2 \times 2 \) table:

\begin{center}
  \begin{tabular}{rll}
     & \scshape lax (mamouras)
     & \scshape regular (us, alvarez-picallo)
    \\
    \scshape function
     & stream transduction
     & stream function
    / differential map
    \\
    \scshape state machine
     & stream transducer
     & \strong{stream processor}
  \end{tabular}
\end{center}

\noindent
We set stream processor in \strong{bold} to indicate it is, to the best of our
knowledge, a novel definition. However, it strongly resembles
\emph{cache-transfer
  style}~\citep{DBLP:conf/esop/GiarrussoRS19,DBLP:conf/pepm/LiuT95}, an
implementation technique for incremental computation which associates each
incrementalized operation with a \emph{cache} maintained as the program's input
changes, analogous to a stream processor's state.
These connections hint at a potential unification of stream processing and
incremental computation; we leave this to future work.

We finally note that we define feedback composition differently from Mamouras.
Mamouras defines feedback on transducers \( \textrm{STrans}(A \dprod B, B) \).
On the other hand, our loop operator takes in a stream processor \(
\processor{} : M \dprod U \leadsto N \dprod U \). Our definition differs for
two reasons: first, our definition allows the right tightening rule
(\propref{prop:right-tightening}). Second extra feedback parameter \( U \) is
because (as future work) we would like to find conditions on \( U \) (but not
\( N \)!) to characterize when fully asynchronous fixed point exists.

\subsection{Flo}

\Citet{DBLP:journals/pacmpl/LaddadCHM25} identify two crucial properties of streaming systems---\emph{eager execution}~\citetext{p.~9, defn.~3.1}%
\ and \emph{streaming progress}~\citetext{p.~10, defn.~3.3}%
---and present a language, Flo, for composing programs out of dataflow operators with these properties.
Unlike the related work previously considered, Flo's definitions are formulated in terms of small-step operational semantics, which complicates direct comparison with our formalism.

\emph{Eager execution} requires a stream program's cumulative output depends only on its cumulative input, not on how that input is ``batched'' or in what order the program is scheduled.
Because we do not give a small-step semantics, we cannot be as flexible in expressing schedulings as Flo.
Determinism across different batchings, however, is implied by our \cref{eqn:derivative} and by the fact that any stream processor implements a (deterministic) stream function.

\emph{Streaming progress} means that a stream program is non-blocking: explicitly terminating an input stream must not affect the output stream other than to terminate it.
For instance, waiting to receive an entire list, then yielding its sum does not
satisfy streaming progress; but prefix sum, which on every new input produces
an updated cumulative sum, does.\footnote{Some operations deliberately block;
  for instance, windowed operators wait for a window to fill. For this reason Flo
  distinguishes a class of \emph{bounded} streams, expected to terminate in
  bounded time, on which blocking is allowed.} Explicit termination is naturally
modeled by monoids with a left zero \( 0 \) such that \( 0 \cdot x = 0 \) (this
is the stream ``terminator''). Streaming progress is then preservation of left
zeroes, \( f(x \cdot 0) = f(x) \cdot 0 \), or if \( f \) is a homomorphism,
simply \( f(0) = 0 \). Generalizing to ticked streams, we could say that a
stream processor \( \sigma : \Tick[M] \leadsto \Tick[N] = (S, f, s_\id, o_\id)
\), satisfies streaming progress if it preserves ticks, meaning \( \out
(f(\top)(s)) = \top \). Unlike Flo we do not require all stream types to
support explicit termination (have a left zero), nor do we give a type system
to guarantee streaming progress; but it seems likely that Flo's
bounded/unbounded stream types could be easily adapted to our setting.%

\subsection{DBSP}

DBSP~\citep{DBLP:journals/vldb/BudiuRZPSKGBCMT25,DBLP:journals/pvldb/BudiuCMRT23}
is a general approach to incremental view maintenance by means of stream
processing. It presents a language of operators on streams over commutative
groups, and an automatic incrementalization strategy for programs in this
language. For instance, we may take a database query \( Q \), lift it into a
DBSP program \( {\uparrow}Q \) over streams of database \emph{snapshots,} then
incrementalize this into a program \( ({\uparrow} Q)^\Delta \) that takes a
stream of database \emph{changes} and produces a stream of query result
changes.

Our formalism always considers streams of changes rather than streams of
snapshots, and so our stream processors correspond roughly to incrementalized
DBSP programs. If we take a stream function \( F \), then lift and
incrementalize it, the result \( ({\uparrow} F)^\Delta \) is effectively a
stream processor \( \sigma \) with \( \denote{\sigma} = F \). However, the
formalisms differ in several ways.
Most significantly, we allow arbitrary monoids where DBSP requires commutative
groups. Group inverses permit a very slick definition of
\emph{incrementalization} in terms of differentiation and integration of
streams, and as the DBSP paper shows, commutative groups are highly expressive.
Nonetheless there are many streaming programs whose ``natural'' input/output
types are not groups. For instance, lists under append, perhaps the most basic
notion of ``stream,'' are a non-commutative monoid.

Smaller differences include that DBSP does not make the state space of a
program explicit; rather, all state must be captured through use of explicit
delay and feedback operators. DBSP also has an explicit notion of time; it
defines streams \( \mathcal{S}_A \) as countable sequences, \( \mathbb{N} \to A
\). By default, there is no notion of time in our formalism; one must
explicitly model time using ticks (\cref{sec:ticks}).

\subsection{Stream Types \& Delta}

\Citet{cutlerstreamtypes} define \emph{Stream Types}, a type system for
\emph{Delta}, a prototype programming language where well-typed programs
implement determinstic stream programs that transform ordered sequences (i.e.
lists) into other ordered sequences. Their type system ensures (1) that
programs emit the same output even when inputs arrive nondeterministically
(e.g. when consuming two parallel input streams) and (2) that programs
statically declare how much intermediate state they need.

Their key technical contribution is their \emph{homomorphism theorem}, which
implies that programs written in Delta (a) output the same result regardless of
the factorization of the input, and (b) produce outputs which can be
appropriately updated when new inputs arrive. Assuming that Delta programs can
be modeled by mathematical functions \( D : \List[T] \to \List[T'] \), we
interpret (a) as saying that \( D \) is a well-defined function, and (b) as
saying that \( D \) satisifes condition~(1) in \cref{def:stream-function}.
Under that framing, well-typed programs in Delta are stream functions over
lists.

\subsection{Rewrite-based optimizations for Paxos}

\Citet{chuoptimizing} demonstrate that domain-specific rewrites can optimize a
naïve implementation of Paxos to compete with hand-written, state-of-the-art
implementations. Their rewrite rules target Datalog with stratified negation
(Dedalus, \cite{dedalus}). Some of their optimizations resemble our
optimizations (specifically, decoupling (\cref{cor:decoupling}) and
partitioning (\cref{lem:independentmerge})). However, in general their
optimizations and proofs of correctness are more complicated than ours because
they also want to ensure their optimizations are correct in real-world
conditions (which may have unreliable networks or machine failures).

\subsection{Loop parallelization}

In the compilers community, much work has been done on parallelizing
sequential, imperative loops. One recent example is
\citet{rompf2017functional}, which approaches the problem of user-defined
aggregation kernels. They construct a language in which the all computations
can be decomposed into a map and reduce phase using bounded static variation
powered by staged computation. Our work focuses on the setting of stream
processors, but both use the idea of the State monoid to represent stateful
computations (though they do not call it as such). Their approach inspired our
use of bounded static variation in \cref{sec:partial-eval} to construct
efficient representations of stream processors.

\subsection{Arrows}

Many of our definitions for composing stream functions and stream processors
are inspired by the \textit{Arrow} typeclass in Haskell, first proposed by
\citet{hughesarrows}. Indeed, Hughes himself noticed that a certain class of
stream programs (that is, possibly-stateful programs that transform an ordered
sequence of input messages into an ordered sequence of output messages) could
be modeled using the Arrow typeclass. We generalize his construction from
ordered sequences to monoids.

\section{Conclusion}

We have demonstrated a general framework for describing, implementing, and
reasoning about stream programs by treating them as homomorphisms into a state
monoid.
We suspect that our approach is suitable for mechanisation in a theorem prover;
in future we would like to pursue this as a first step towards a verified
optimizing compiler for stream programs. We would also like to extend our
approach to other semantic constructs, such as \emph{asynchronous} feedback
(feedback \emph{without} delay, subject to conditions ensuring a fixed point
exists, \`a la Datalog). Most speculatively, we believe it may be possible to
use the compactness theorem from first-order logic to show that any function
satisfying \cref{eqn:derivative} can be extended to one also satisfying the
(more technical) conditions \crefrange{eqn:regular-id}{eqn:regular-mul}, which
would fully generalize our definition of stream function.

\bibliography{monoids}

\end{document}